\documentclass[10pt,journal,compsoc]{IEEEtran}

\usepackage{algorithm,algorithmic,amsbsy,amsmath,amssymb,epsfig,bbm,mathrsfs,multirow,amsthm,epstopdf,url,amsfonts,textcomp,xcolor}
\usepackage{graphicx,caption,subcaption}
\usepackage{setspace}

\ifCLASSOPTIONcompsoc
  \usepackage[nocompress]{cite}
\else
  \usepackage{cite}
\fi
\ifCLASSINFOpdf
\else
\fi

\newtheorem{lemma}{Lemma}
\newtheorem{theorem}{\textbf{\textsc{Theorem}}}
\newcommand{\defeq}{\stackrel{\mbox{{\tiny def}}}{=}}

\begin{document}
\title{A Novel Mobile Edge Network Architecture with Joint Caching-Delivering and Horizontal Cooperation}
\author{Yuris Mulya Saputra, Dinh Thai Hoang, Diep N. Nguyen, and Eryk Dutkiewicz
\IEEEcompsocitemizethanks{\IEEEcompsocthanksitem The authors are with the School of Electrical and Data Engineering, University of Technology Sydney, Australia.\protect\\
E-mail: yurismulya.saputra@student.uts.edu.au, \hfil\break \{hoang.dinh, diep.nguyen, eryk.dutkiewicz\}@uts.edu.au.
}
\thanks{}}

\IEEEtitleabstractindextext{%
\begin{abstract}
Mobile edge caching/computing (MEC) has been emerging as a promising paradigm to provide ultra-high rate, ultra-reliable, and/or low-latency communications in future wireless networks. In this paper, we introduce a novel MEC network architecture that leverages the optimal joint caching-delivering with horizontal cooperation among mobile edge nodes (MENs). To that end, we first formulate the content-access delay minimization problem by jointly optimizing the content caching and delivering decisions under various network constraints (e.g., network topology, storage capacity and users' demands at each MEN). However, such strongly mutual dependency between the decisions creates a nested dual optimization that is proved to be NP-hard. To deal with it, we propose a novel transformation method to transform the nested dual problem to an equivalent mixed-integer nonlinear programming (MINLP) optimization problem. Then, we design a centralized solution using an improved branch-and-bound algorithm with the interior-point method to find the joint caching and delivering policy which is within 1\% of the optimal solution for small problem instances. Since the centralized solution requires the full network topology and information from all MENs, to make our solution scalable, we develop a distributed algorithm which allows each MEN to make its own decisions based on its local observations. Extensive simulations demonstrate that the proposed solutions can reduce the total average delay for the whole network up to 40\% compared with other current caching policies. Furthermore, the proposed solutions also increase the cache hit rate for the network by four times, thereby dramatically reducing the traffic load on the backhaul network.

\end{abstract}

}

\maketitle

\IEEEdisplaynontitleabstractindextext

\IEEEpeerreviewmaketitle

\IEEEraisesectionheading{\section{Introduction}\label{sec:Int}}

\IEEEPARstart{T}{he} proliferation of smart devices (e.g., smartphones, tablets, wearable devices) and a massive demand for emerging services (e.g., IoT, mobile videos, augmented and virtual reality applications) will lead to an explosion of mobile data traffic in the near future~\cite{Cisco:2017}. Moreover, new services/applications (e.g., critical missions, safety, or real-time entertainment) often require much lower latency (e.g., in order of 10 milliseconds or less) and/or higher reliability (e.g., 99.999\%).  
To cope with this ever-increasing traffic demand and much more stricter quality of service, mobile edge caching/computing (MEC)~\cite{Zhang:2015,Mao:2017,Yang:2016,Li2:2018} has been proposed as one of the most effective solutions. The key idea of an MEC network is to distribute well-received contents as well as computing resources closer to the mobile users by deploying servers at the ``edge'' of the network, referred to as mobile edge nodes/servers (MENs). By doing so, the MEC network allows mobile users, instead of downloading contents from cloud/content servers (CSs), to access contents and resources from nearby MENs. As such, MEC helps to reduce the service latency and mitigate the network congestion on the backhaul link~\cite{Li:2018}. The deployment of MEC network also brings other significant benefits to the mobile users, e.g., reliable wireless connections, high speed data transfer, and low energy consumption, and reduces expensive operational as well as upgrading costs on the backhaul link for the MEC network providers.

Despite of the advantages, the development of MEC networks faces several inherent challenges, e.g., diverse users' demands, small coverage, and limited storage capacity of each MEN~\cite{Hoang:2018}. To overcome these issues, cooperative caching has been introduced recently. Cooperative caching is a method which leverages the distribution of contents~\cite{Ioannou:2016} through the cooperation among MENs and provides efficient workload distribution in a hierachical architecture~\cite{Tong:2016}. In~\cite{Wen:2017}, the authors introduce an optimal cooperative content placement strategy aiming to maximize the total hit rate of contents for a heterogeneous cellular network (referred to as the network hit rate). To achieve the optimal policy, the scheme first caches contents randomly at different layers, e.g., macro, pico, and femto nodes. Then, a stochastic geometry is deployed to evaluate the request hit rate/probability for the network. Based on the request hit rate evaluation together with the caching storage capacities of MENs, the maximum network hit rate then can be derived. Similar to~\cite{Ioannou:2016}, the authors in~\cite{Poularakis:2016} propose a distributed cooperative caching architecture consisting of several local nodes such as small cells or base stations to reduce the content delivery delay. In this work, the network providers collaborate together through sharing their MENs to avoid long service delay from the content servers. For example, users of network provider X can download contents from MENs of network provider Y. In~\cite{Borst2010Distributed}, a collaborative caching design based on a tree architecture is considered. In this way, if a requested content is not cached at the leaf node server, this node can download content from its directly connected parent node. Based on this architecture, the authors introduce a two-level cooperative caching group model to minimize the total bandwidth cost for the system. The work is then extended in~\cite{Shanmugam2013Femtocaching} with network topology taken into account.

To minimize the service delay, most of the above cooperative caching solutions focus only on content placement strategies (e.g., \cite{Li:2018}, \cite{Borst2010Distributed}, and \cite{Shanmugam2013Femtocaching}), but not on how content should be delivered/routed. Among few works that investigate both caching and delivering strategies, the authors in~\cite{Jiang:2017} introduce dual objective approaches to optimize cooperative caching and delivering policy for heterogeneous cellular networks (HetNets) with femtocells and D2D communications. As such, after obtaining the optimal caching strategy, the authors adopt an optimal content delivering strategy to find the best MEN to serve if there are more than one MENs caching the same requested content. Nonetheless, in this work, caching placement and delivery optimizations are separately optimized, leading to a suboptimal solution. A similar suboptimal solution is also found in \cite{Dehghan2015On}.

Moreover, for all the aforementioned works, the authors do not take into account the cooperations among MENs (referred to as \emph{horizontal cooperation}) in delivering requested contents. MENs are often launched in a close proximity where the horizontal (wireless or even wire) connections among them have much higher speed than that of the backhaul link or those from MENs to the CSs. This fact, if leveraged properly, facilitates the cooperation amongst MENs in jointly caching and delivering contents to not only reduce the service delay for mobile users but also improve the caching effectiveness (by minimizing the backhaul traffic to the CSs). 

In this paper, we introduce an optimal {\bf JO}int {\bf C}ooperative c{\bf A}ching and {\bf D}elivering framework (referred to as JOCAD) that enables MENs to cooperate in caching and delivering contents to mobile users. Specifically, we first propose a novel MEC network architecture in which MENs can be connected with each other directly (with high-speed connections). Then, given the content demand distributions (referred to as frequency-of-accesses) at MENs, various data sizes, diverse MENs' storage capacities, and network topology (i.e., connections among the MENs and their bandwidth), we aim to jointly address two essential questions: (1) how to place contents at the MENs efficiently and (2) how to choose the best routes to deliver requested contents. 

Due to the strongly mutual dependency between content placement and delivery decisions, the joint content caching and delivering optimization problem turns to be intractable. In particular, where to cache contents will be influenced by the delivering decisions. Likewise, the decision to deliver the contents will be impacted by the content placement strategy. Consequently, this inter-dependency gives rise to a nested dual optimization problem, that we prove to be NP-hard. To tackle it, we propose a novel transformation method to transform the nested dual problem to an equivalent mixed-integer nonlinear programming (MINLP) optimization problem. By exploiting the unique structure of this MINLP problem, we propose a centralized cooperative caching-delivering solution (referred to as centralized solution) using an improved branch-and-bound algorithm with the interior-point method (iBBA-IPM) to find the joint caching and delivering policy which is within 1\% of the optimal solution for small problem instances. 

Nevertheless, the centralized solution requires a full network topology and the information from all MENs, making it prohibitively costly for large-scale systems, especially when the number of contents and MENs is huge. Moreover, the centralized solution also requires coordination overhead from all MENs. To make our solution scalable and reduce the complexity and communication overheads among MENs, we develop a distributed cooperative caching{-delivering} algorithm (referred to as distributed solution) which allows each MEN to make its own decisions based on its local observations, e.g., connections to their neighbors. Extensive simulations demonstrate that the proposed solutions can reduce the total average delay for the whole network up to 40\% compared with the most FoA policy, and up to 25\% compared with locally optimal caching policy (i.e., without collaboration). Furthermore, the proposed solutions also increase the cache hit rate for the network by four times, thereby dramatically reducing the traffic load on the backhaul network.The major contributions are summarized as follows:

\begin{itemize}
\item We design the JOCAD framework with mutual dependency that utilizes the direct horizontal cooperations among MENs to minimize the total delay for the MEC network and reduce traffic load on the backhaul links. We show that the resulting optimization problem is NP-hard.

\item We then propose a novel transformation method to transform the intractable original optimization to an equivalent MINLP problem which we can adopt effective mathematical tools to address.

\item We develop the centralized solution using iBBA-IPM to find the joint caching and delivering policy for the MINLP problem which can achieve within 1\% of the optimal solution. 

\item We design the distributed solution to reduce the complexity of the centralized solution. Through simulations, we {demonstrate} that the distributed solution can achieve the performance close to that of the centralized solution.

\item We conduct extensive simulations to evaluate the efficiency of the proposed framework and solutions. These results also provide insightful information to help the MEC service providers tradeoff between the quality of service, e.g., delay, and the implementation costs, e.g., the number of MENs and storage capacity. 
\end{itemize}

The rest of this paper is organized as follows. Section~\ref{sec:SM} describes the proposed network architecture and model. Section~\ref{sec:PFT} discusses the problem formulation and the transformation method. The centralized and distributed cooperative caching{-delivering} solutions are in Section \ref{sec:PS} and Section \ref{sec:DS}, respectively. An illustrative case study is given in Section 6, and then Section 7 shows the comprehensive simulation results. Conclusions are then drawn in Section~\ref{sec:Conc}.

\begin{figure}[!t]
	\centering
	\includegraphics[scale=0.3]{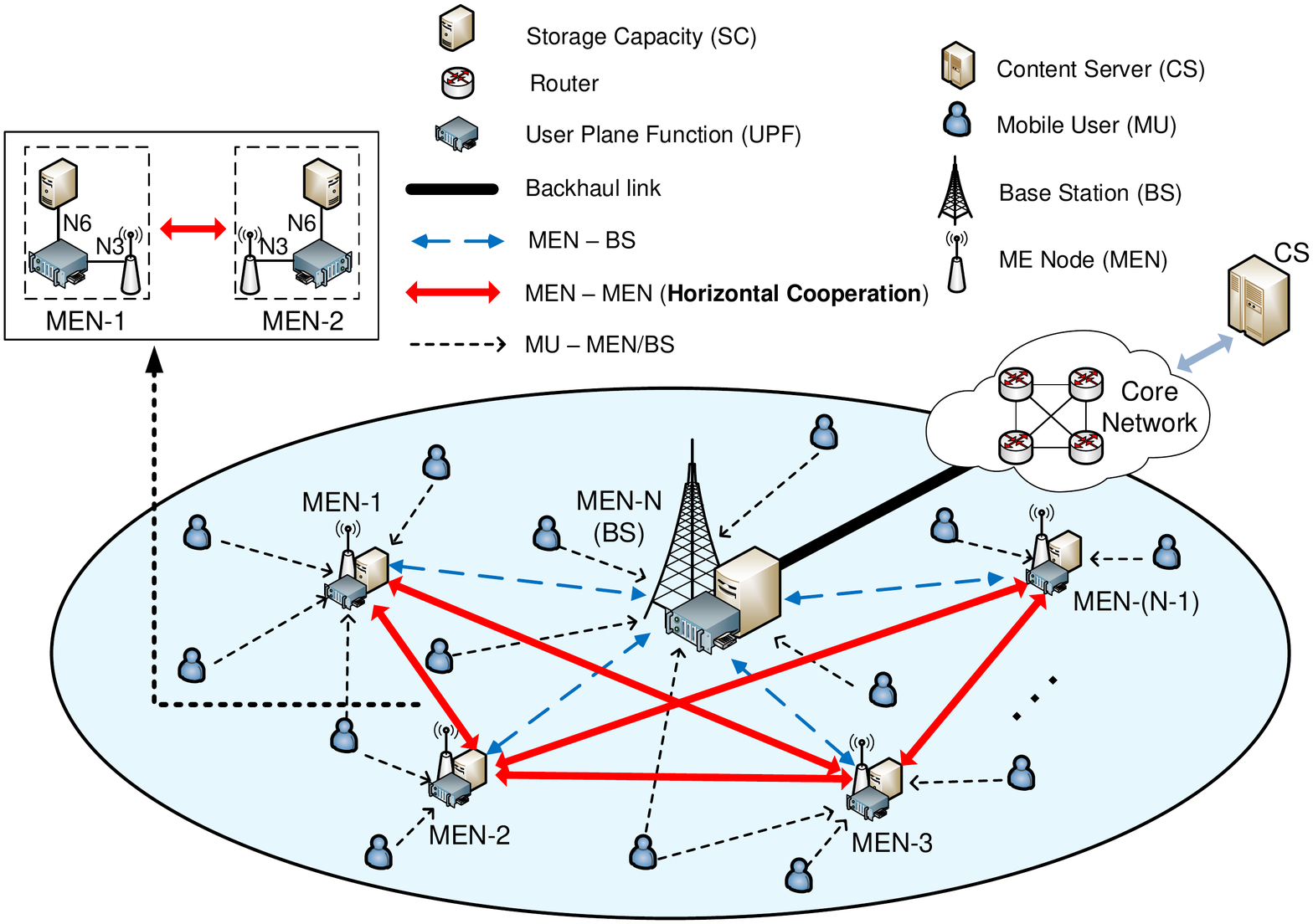}
	\caption{Proposed MEC network architecture with direct horizontal cooperation.}
	\label{fig:Edge_Architecture}
\end{figure}

\section{Mobile Edge Network Architecture with Horizontal Cooperation amongst MENs}
\label{sec:SM}
	
The proposed MEC network architecture with direct horizontal collaboration among the MENs is described in Fig.~\ref{fig:Edge_Architecture}. This architecture adopts a practical MEC model in 5G networks proposed by The European Telecommunications Standards Institute (ETSI)~\cite{Kekki:2018}. Each MEN in the network serves a set of mobile users in its coverage area and it can communicate with its nearby MENs (the BS is also considered as an MEN in the network) using either wireless (e.g., Wi-Fi) or wire connections. Additionally, each MEN is equipped with finite storage capacity to cache popular contents. In this case, each MEN can implement periodic updating strategies of the high-demand content placement within a day at off-peak hours, e.g., at night, to reduce traffic loads during peak hours~\cite{Yang:2016, Karamchandani:2016}. To route data traffic, e.g., content requests and content downloads, from/to the mobile users or other MENs, each MEN also utilizes a user plane function (UPF). This UPF is connected to the MEN and the storage capacity via N3 and N6 interfaces, respectively~\cite{Kekki:2018}. Moreover, the mobile users can move from one MEN's coverage to another MEN's area to create connections for content requests. 

When a content request is sent to an MEN, if the content is cached locally at the MEN, it will send the content to the user immediately through the UPF. If the content is not stored locally at the MEN but at its neighboring MENs (i.e., directly connected to the MEN), the MEN will download the content from the node which has the lowest delivery delay before sending the requested content to the user. Otherwise, the MEN will download the content from the CS via the BS{\footnote{In practice, the bandwidth between the BS and CS is usually much higher than among MENs \cite{Poularakis2:2016}. Thus, downloading the content from the CS via the BS will have lower delay than getting it through two MEN-MEN hops.}}. In this case, the BS requires to download the requested content from the CS through the 5G core network (CN)~\cite{Li:2017}. Note that if the MEN is the BS, it will check from MENs in the network first, and if there is no MEN storing this content, it will download the content from the CS via the 5G CN. In this way, the proposed model can leverage the direct horizontal cooperations among MENs to reduce content-access delay for the users as well as traffic load on the backhaul network. 

Let $\mathcal{N} = \{1,\ldots,n,\ldots,N\}$ denote the set of the MENs. The BS is denoted by MEN-$N$. Each MEN-$n$ has storage capacity and the number of mobile users in its coverage area denoted by $s_n$ and $U_n$, respectively. $B$ represents the bandwidth between the BS and the CS. We also define $l^m_n, \forall n,m \in \mathcal{N}, n\neq m$ as the bandwidth between MEN-$n$ (note that $B\gg l^m_n$) and MEN-$m$, and $b_n^u, \forall n \in \mathcal{N}$ as the allocated bandwidth of a user $u$ at MEN-$n$. Furthermore, we define $\mathcal{I} = \{1,\ldots,i,\ldots,I\}$ as the set of contents. Contents may have diverse data sizes denoted by $\mathcal{C} = \{c_1,\ldots,c_i,\ldots,c_I\}$. We denote frequency-of-access (FoA) of content $i$ at MEN-$n$ as $f_n^i$.

\section{Joint Cooperative Caching and Delivering Optimization Problem} \label{sec:PFT}
In this section, we first formulate the joint cooperative caching and delivering optimization problem to minimize the total average delay for the MEC network under various network constraints as aforementioned. This nested dual optimization problem is shown to be NP-hard. We then propose a method to transform the intractable nested dual optimization problem into the equivalent MINLP for which we can develop effective solutions. 

\subsection{Decision Variables and Problem Analysis}
\label{subsec:PF}

We define $\mathbf{x} \defeq [\mathbf{x}_1,\ldots,\mathbf{x}_n,\ldots,\mathbf{x}_N]^T$, where $\mathbf{x}_n = [x_n^1,\ldots,x_n^i,\ldots,x_n^I] $ and $x_n^i \in \{0,1\}$, as the binary decision vector of the MENs. The variable $x_n^i$ is defined as follows:
\begin{equation}
\label{eqn:prop_cond0}
\begin{aligned}
x_n^i = \quad
\left\{	\begin{array}{ll}
1, &\text{ if content $i$ is cached at MEN-$n$}, \\
0, &\text{ otherwise}.
\end{array}	\right.\\
\end{aligned}
\end{equation}
When a user $u$ at MEN-$n$ requests a content $i$, the following three cases are considered.

\subsubsection{Case 1: MEN-$n$ has the requested content} 
In this case, $x^i_n=1$, and thus the delay to download the content will be 
\begin{equation}
\label{eqn:prop_cond1}
\begin{aligned}
d_{case1}^* \defeq x^i_n d_\alpha,
\end{aligned}
\end{equation}
where $d_\alpha = \frac{c_i}{b_n^u}$.

\subsubsection{Case 2: MEN-$n$ does not have the requested content but at least one of its directly connected MENs caches this content} 
In this case, $x^i_n=0$. We denote $\mathcal{M}_n^i$ as the set of MENs which are directly connected to MEN-$n$ and store content $i$. Because there is at least one MEN in $\mathcal{M}_n^i$ caching the content $i$, we can derive:
\begin{equation}
\label{eqn:prop_cond2}
\begin{aligned}
x_{m^\dagger}^i\defeq\underset{{\substack{m\in \mathcal{M}_n^i}}}{\prod}(1 - x_{m}^i) = 0,
\end{aligned}
\end{equation}
where 
\begin{equation}
\label{eqn:prop_cond02}
\begin{aligned}
x_m^i = \quad
\left\{	\begin{array}{ll}
1, &\text{ if content $i$ is cached at MEN-$m$}, \\
   &\text{ $m \neq n$},\\
0, &\text{ otherwise}.
\end{array}	\right.\\
\end{aligned}
\end{equation}
As such, the MEN-$n$ will download content $i$ from one node which has the lowest delivery time, and thus the delay to download the content in this case will be:
\begin{equation}
\label{eqn:prop_cond3}
\begin{aligned}
&d_{case2}^* \defeq (1-x^i_n) (1-x_{m^\dagger}^i) (d_\alpha + d_\beta),
\end{aligned}
\end{equation}
where $d_\beta \defeq \underset{m}{\min}\Big([x_m^i + (1 - x_m^i)V] \frac{c_i}{l^m_n}\Big)$, and $V$ is a very large constant number. Here, $d_\beta$ represents the optimization problem to find an MEN-$m$ which contains requested content $i$ and has the lowest delivery time. In addition, the $V$ ensures that MENs without containing the requested content $i$ (i.e., $x_m^i=0$) are practically ignored. 

It is worth noting that our model is designed to take the full advantage of direct horizontal cooperations among MENs. In fact, MENs are often deployed in a close proximity area. As such, the wire or wireless links among MENs are often in place \cite{Poularakis:2016,Ao:2015,Poularakis:2014} and usually much faster than the links from MENs to the BS. In particular, if one of its directly connected MENs stores content $i$, MEN-$n$ will download this content from one of these nodes instead of trying to download the content from the CS or from another MEN in the network via the BS. This strategy is to minimize the latency, the traffic on the backhaul network as well as inside the MEC network. 

\begin{figure*}[t]
	\normalsize
	\begin{equation}
	\label{eqn:form_problem1a0}	
	\begin{aligned}
	F(\mathbf{x}) &= \sum_{i=1}^I \sum_{n=1}^{N} f_n^i \Bigg[ \sum_{u=1}^{U_n} \bigg( d_{case1}^* + d_{case2}^* + d_{case3}^* \bigg)\Bigg] \\
	&= \sum_{i=1}^I \sum_{n=1}^{N} f_n^i \Bigg[ \sum_{u=1}^{U_n} \bigg( \underbrace{x_n^id_\alpha}_{\text{ \emph{Case 1}}} + \underbrace{ (1-x_n^i) (1 - x_{m^\dagger}^i) \bigg(d_{\alpha} + \min_{m}\Big([x_m^i + (1 - x_m^i)V] \frac{c_i}{l^m_n}\Big) \bigg)}_{\text{ \emph{Case 2}}} + \underbrace{(1-x_n^i) x_{m^\dagger}^i d_{\delta}}_{\text{ \emph{Case 3}}} \bigg)	\Bigg] \\
	&= \sum_{i=1}^I \sum_{n=1}^{N} f_n^i \Bigg[ \sum_{u=1}^{U_n} \bigg( x_n^i d_\alpha + (1 - x_n^i) \Big[(1 - x_{m^\dagger}^i) (d_\alpha + d_\beta) + x_{m^\dagger}^id_{\delta}\Big]\bigg)\Bigg] \\
	\end{aligned}
	\end{equation}
	\hrulefill
	\vspace*{-4pt}
\end{figure*}

Note that in the case MEN-$n$ is the BS, i.e., $n=N$, only \emph{Case 1} and \emph{Case 2} can happen as the BS is directly connected to the CS. Furthermore, we do not consider direct connections between MENs and the CS. This is due to the fact that the direct connection between an MEN and the CS often has a very low bandwidth capacity \cite{Poularakis2:2016}. Another reason is that deploying direct connections between MENs and the CS will cause significant infrastructure deployment and operational costs for the MEC service providers \cite{Dehghan2015On, Ao:2015, Ge:2014, Zhang:2018}.

\subsubsection{Case 3: MEN-$n$ and all of its directly connected nodes do not have the requested content}

In this case, we have $x^i_n=0$ and $x_{m^\dagger}^i=1$. Then, there are two possibilities. First, there is no MEN in the network storing the content. Then, the content will be downloaded from the CS via the BS, and the delay in this case will be 
\begin{equation}
\label{eqn:prop_cond4}
\begin{aligned}
&d_{case3}^* \defeq (1-x_n^i) x_{m^\dagger}^i d_{\delta},
\end{aligned}
\end{equation}
where $d_{\delta} = \Big(d_\alpha + \frac{c_i}{l^N_n} + \frac{c_i}{B}\Big)$. Second, if there is at least one MEN in the network that is not directly connected to MEN-$n$, say MEN-$m$ (with some abusing of notation), storing this content, MEN-$n$ will download the content from either the CS via the BS or from MEN-$m$ via the BS or any intermediate MEN (whichever has lower delivery delay). However, in practice, as aforementioned the bandwidth between the BS and CS is usually much higher than among MENs \cite{Poularakis2:2016}, and thus the content will be downloaded from the CS in the second case. Consequently, the delay to download the content in \emph{Case 3} will be $d_{\delta}$.

\subsection{Problem Formulation}
Based on the aforementioned analysis, we then formulate the joint cooperative caching and delivering optimization problem ($\mathbf{P}_1$) to minimize the total average delay of the MEC network as follows:
\begin{equation}
\label{eqn:form_problem1a}	
\begin{aligned}
(\mathbf{P}_1) \phantom{10} & \min_{\mathbf{x}} F(\mathbf{x}),
\end{aligned}
\end{equation}
\begin{eqnarray}
\text{s.t.} \quad \sum_{i=1}^Ix_n^ic_i \leq s_n, \forall n \in \mathcal{N},  \label{eqn:form_problem1b} \\
x_n^i, x_m^i \in \{0, 1\}, \forall n,m \in \mathcal{N},  n \neq m, \forall i \in \mathcal{I}, \label{eqn:form_problem1c}
\end{eqnarray}
{where the objective function $F(\mathbf{x})$ is defined in Eq.~(\ref{eqn:form_problem1a0}) as the total average delay at all MENs for all contents in the MEC network}.
The constraints~(\ref{eqn:form_problem1b}) guarantee that the total size of all cached contents does not exceed the storage capacity of each MEN. Additionally, the constraints~(\ref{eqn:form_problem1c}) specify that caching decision variables are binary. Based on $(\mathbf{P}_1)$, our optimization problem is considered to be a \emph{nested dual binary nonlinear programming}. In particular, binary decision variables are multiplied and used in both inner and outer minimization functions: (1) the outer level (OL) is the main objective function $F(\mathbf{x})$ to minimize the total average delay of the network and (2) the inner level (IL) is the optimal delivering decision problem, i.e.,
\begin{equation}
\label{eqn:form_problem1e}
\begin{aligned}	
F_{\emph{\mbox{IL}}}(\mathbf{x}) = d_\beta = \underset{m}{\min}\Big([x_m^i + (1 - x_m^i)V] \frac{c_i}{l^m_n}\Big).
\end{aligned}
\end{equation}
The IL will try to find a directly connected node, i.e., MEN-$m$, which minimizes the delivery time to a requesting mobile user as described in \emph{Case 2}. Furthermore, the OL and IL functions have a unique feature, i.e., \emph{mutual-dependency}. Specifically, on the one hand, where to cache contents (at MENs) will influence on how to deliver the contents. On the other hand, how to deliver contents (depending on the network topology) will impact on how/where contents should be cached at MENs. In this way, to address the OL problem, we need to solve the IL problem by obtaining the optimal value of $m$, i.e., which MENs to deliver. Furthermore, it is important to note that we cannot solve the IL problem to find the optimal values of $m$ separately because the optimal values of $m$ strongly depend on the variables $\mathbf{x}$ (i.e., which contents are cached at which MENs). Consequently, to minimize the total average delay for the network, we need to simultaneously address the joint caching and delivering optimization problem. In Lemma~\ref{lemma1}, we show that the optimization problem in Eq.~(\ref{eqn:form_problem1a}) is an NP-hard problem.

\begin{lemma}
\label{lemma1}
The nested dual optimization ($\mathbf{P}_1$) is NP-hard.
\end{lemma}
 
\begin{proof}
See Appendix~\ref{appx:lemma1}.
\end{proof}
	
\subsection{Problem Transformation}
\label{subsec:PT}

In ($\mathbf{P}_1$), we need to simultaneously optimize the content caching and delivering decisions. However, as discussed above, the decisions to cache and to deliver the contents are mutual dependent. Hence, the optimization problem is unfortunately not a typical bilevel optimization \cite{Sinha:2018} in which the inner and outer decision variables are independent. As such, conventional techniques used in the bilevel optimization are not applicable. To address this problem, we propose a novel method to transform the intractable nested dual optimization problem into an equivalent mixed-integer nonlinear programming (MINLP) optimization problem which can be solved by using effective mathematical tools.

Specifically, we introduce a vector of functions $Q(\mathbf{x})$ as well as vectors of auxiliary variables $\mathbf{z}$ and binary variables $\mathbf{y}$ to replace the IL problem $F_{\emph{\mbox{IL}}}(\mathbf{x})$, and add appropriate constraints to make the transformed problem to be equivalent. In this case, $Q(\mathbf{x})$ represents the IL problem. Particularly, $Q(\mathbf{x}) \defeq [Q_n(\mathbf{x}_1),\ldots,Q_n(\mathbf{x}_m),\ldots,Q_n(\mathbf{x}_{M_n})]^T$, where $Q_n(\mathbf{x}_m) \defeq [Q_n^1(x_m^1),\ldots,Q_n^i(x_m^i),\ldots,Q_n^I(x_m^I)]$ and $Q_n^i(x_m^i) \defeq [x_m^i + (1 - x_m^i)V] \frac{c_i}{l^m_n}$. Then, we define $\mathbf{z} = [\,\mathbf{z}_1,\ldots,\mathbf{z}_n,\ldots,\mathbf{z}_N]\,^T$ with $\mathbf{z}_n = [\,z_n^1,\ldots,z_n^i,\ldots,z_n^I]\,$, where $z_n^i \in \mathbb{R}_0^+$, as the vector of $F_{\emph{\mbox{IL}}}(\mathbf{x})$ solutions for all contents and MENs. Furthermore, we denote $\mathbf{y} = [\,\mathbf{y}_1,\ldots,\mathbf{y}_m,\ldots,\mathbf{y}_N]\,^T$ with $\mathbf{y}_m = [\,y_m^1,\ldots,y_m^i,\ldots,y_m^I]\, , y_m^i \in \{0,1\},$ and $m \neq n$, as the binary decision vector to help MEN-$n$, $\forall n \in \mathcal{N}$ to find another directly connected MEN-$m$, $m \in \mathcal{M}_n^i$, which has the shortest time to deliver requested content $i$.

Given the vectors $Q(\mathbf{x})$, $\mathbf{z}$, and $\mathbf{y}$, the new equivalent optimization problem ($\mathbf{P}_2$), known as MINLP problem, can be expressed as follows:
\begin{equation}
(\mathbf{P}_2) \phantom{5} \underset{\{{\mathbf{x},\mathbf{y},\mathbf{z}}\}}{\min} F(\mathbf{x},\mathbf{y},\mathbf{z}), \label{eqn:prop_sol2a}
\end{equation}
\begin{eqnarray}
\text{ s.t. (\ref{eqn:form_problem1b})-(\ref{eqn:form_problem1c}) and} \nonumber \\ 
z_n^i \ge Q_n^i(x_{m}^i) - V y_{m}^i, \forall n, m \in \mathcal{N}, m \neq n, \forall i \in \mathcal{I}, \label{eqn:prop_sol2c} \\
\sum_{m \in \mathcal{M}_n^i} y_{m}^i = M_n^i-1, \forall n \in \mathcal{N}, \forall i \in \mathcal{I},\label{eqn:prop_sol2d} \\
y_{m}^i \in \{0, 1\}, \forall m \in \mathcal{N}, \forall i \in \mathcal{I}, \label{eqn:prop_sol2e} \\
z_n^i \in \mathbb{R}_0^+, \forall n \in \mathcal{N}, \forall i \in \mathcal{I}, \label{eqn:prop_sol2f} 
\end{eqnarray}
where 
\begin{equation}
\label{eqn:prop_sol2a0}
\begin{aligned}
F(\mathbf{x},\mathbf{y},\mathbf{z})  &= \sum_{i=1}^I \sum_{n=1}^{N} f_n^i \Bigg[ \sum_{u=1}^{U_n} \bigg( x_n^i d_\alpha + \\
&(1 - x_n^i) \Big[(1 - x_{m^\dagger}^i) (d_\alpha + z_n^i) + x_{m^\dagger}^id_{\delta}\Big]\bigg)\Bigg],\\
\end{aligned}
\end{equation} 
and $M_n^i$ is the cardinality of the set $\mathcal{M}_n^i$. The constraints (\ref{eqn:prop_sol2c}) represent the condition to select one MEN-$m$ which has the lowest delivery time. In addition, the aim of constraints (\ref{eqn:prop_sol2d}) and (\ref{eqn:prop_sol2e}) are to guarantee that only one variable $y_m^i$ is set to be ``0'', while the rest of variables are set to be ``1''. Specifically, $M_n^i-1$ indicates that when $M_n^i$ number of directly connected MENs containing content $i$ for MEN-$n$ are considered, we exclude one node which has $y_m^i = 0$ (i.e., the selected MEN-$m$ to deliver the content $i$). The equivalent transformation is formally stated in Theorem \ref{theorem1}.

\begin{theorem}\label{theorem1}
The nested dual optimization problem ($\mathbf{P}_1$) is equivalent to MINLP optimization problem ($\mathbf{P}_2$).
\end{theorem}
 
\begin{proof}
See Appendix~\ref{appx:theorem1}.
\end{proof}

Based on ($\mathbf{P}_2$), we can simplify some parts of the objective function in Eq.~(\ref{eqn:prop_sol2a}) as
\begin{equation}
\label{eqn:prop_sol20a}
\begin{aligned}
\Omega \big(x_n^i,x_m^i\big) &\defeq  x_n^id_\alpha+ (1-x_n^i) (1 - x_{m^\dagger}^i) d_{\alpha} \\
& + (1-x_n^i) x_{m^\dagger}^i d_{\delta},
\end{aligned}
\end{equation}
\begin{equation}
\label{eqn:prop_sol20b}
\begin{aligned}
\Phi \big(x_n^i,x_m^i\big) \defeq (1-x_n^i) (1 - x_{m^\dagger}^i),
\end{aligned}
\end{equation}
the constraints (\ref{eqn:form_problem1b}) as
\begin{equation}
\label{eqn:prop_sol20c}
\begin{aligned}
\Theta \big(x_n^i\big) \defeq {\underset{i=1}{\overset{I}{\sum}}}x_n^ic_i - s_n,
\end{aligned}
\end{equation}
and the constraints (\ref{eqn:prop_sol2d}) as
\begin{equation}
\label{eqn:prop_sol20d}
\begin{aligned}
\Gamma \big(y_{m}^i\big) \defeq {\underset{{m \in \mathcal{M}_n^i}}{\sum}} y_{m}^i - (M_n^i-1).
\end{aligned}
\end{equation}
{Then, we have}
\begin{equation}
\label{eqn:prop_sol30a0}
\begin{aligned}
&F(\mathbf{x},\mathbf{y},\mathbf{z}) = \\
&\sum_{i=1}^I \sum_{n=1}^{N} f_n^i \Bigg[ \sum_{u=1}^{U_n} \bigg(\Omega \big(x_n^i,x_m^i\big) + \Phi \big(x_n^i,x_m^i\big)z_n^i\bigg)\Bigg],
\end{aligned}
\end{equation}
and the MINLP problem $(\mathbf{P}_2)$ becomes
\begin{equation}
\label{eqn:prop_sol30a}
\begin{aligned}
(\mathbf{P}_3) \phantom{5} & \min_{\{\mathbf{x},\mathbf{y},\mathbf{z}\}} F(\mathbf{x},\mathbf{y},\mathbf{z}),
\end{aligned}
\end{equation}
\begin{eqnarray}
\text{ s.t. (\ref{eqn:form_problem1c}), (\ref{eqn:prop_sol2e})-(\ref{eqn:prop_sol2f}), and} \nonumber \\ 
\Theta \big(x_n^i\big) \leq 0, \forall n \in \mathcal{N}, \label{eqn:prop_sol3b} \\
Q_n^i \big(x_m^i\big) - V y_{m}^i - z_n^i \leq 0, \nonumber \\
\forall n, m \in \mathcal{N}, m \neq n, \forall i \in \mathcal{I}, \label{eqn:prop_sol3d} \\
\Gamma \big(y_{m}^i\big) = 0, \forall n \in \mathcal{N}, \forall i \in \mathcal{I}. \label{eqn:prop_sol3e}
\end{eqnarray}

\section{Centralized Cooperative Caching{-Delivering} Solution}
\label{sec:PS}

To address ($\mathbf{P}_3$), in this section, we introduce an improved branch-and-bound algorithm which can effectively achieve the near-optimal joint caching and delivering policy for the whole network.

\subsection{Improved Branch-and-Bound Algorithm with Interior-Point Method}
\label{sec:BBIPM}

We adopt branch-and-bound algorithm (BBA)~\cite{Belotti:2009} which reduces the time complexity and leverages the characteristics of binary variables to find the final solution of ($\mathbf{P}_3$). The BBA has been commercially popular (in CPLEX by IBM~\cite{IBM:2016} or in SBB by GAMS~\cite{SBB:2009}) and it can work very effectively to solve integer programming problems. Nonetheless, to solve the non-linearity and continuous relaxation of ($\mathbf{P}_3$), we employ the interior-point method (IPM) which has polynomial time complexity~\cite{Karmarkar:1984}. The IPM is considered to be robust and efficient method to address function evaluations and second derivative information for large-scale and sparse nonlinear problem~\cite{Byrd:1999} (as implemented in IPOPT~\cite{Wachter:2006} and in KNITRO~\cite{Byrd:2006}. Overall, both BBA and IPM can be integrated to handle binary and continuous variables of ($\mathbf{P}_3$) efficiently.

Given the aforementioned BBA and IPM, we develop a two-level approach to effectively address ($\mathbf{P}_3$). In particular, the BBA is used as outer-level approach which creates continuous nonlinear subproblems by relaxing binary constraints of ($\mathbf{P}_3$), and then the IPM is applied as inner-level approach to solve the subproblems. In this way, the optimal solutions of subproblems attained from the IPM may not be integral (some variables do not have binary values). Thus, the BBA is then used again to find the feasible (integral) solutions until the final solution of ($\mathbf{P}_3$) is found.    

To find the final solution of ($\mathbf{P}_3$) using the BBA, we recall binary variable vectors $\mathbf{x}$ and $\mathbf{y}$, and non-integer variable vector $\mathbf{z}$. Given $I$ number of contents and $N$ number of MENs, the number of variables of vector $\mathbf{x}$, $\mathbf{y}$, and $\mathbf{z}$ are  $J_{\mathbf{x}} = J_{\mathbf{y}} = J_{\mathbf{z}} = I \times N$. 
The BBA first relaxes all binary variables $x_n^i$, $x_m^i$, and $y_m^i$ of ($\mathbf{P}_3$) into continuous variables at the root problem ($\mathbf{RP}$). In the ($\mathbf{RP}$), the relaxed binary variables are bounded by $0 \leq x_n^i, x_m^i, y_m^i \leq 1$. Then, the ($\mathbf{RP}$) can be expressed as follows:
\begin{equation}
\label{eqn:prop_sol3a}
\begin{aligned}
(\mathbf{RP}) \phantom{5} &\min_{\{\mathbf{x},\mathbf{y},\mathbf{z}\}}  F_{\mathbf{RP}}(\mathbf{x},\mathbf{y},\mathbf{z}),
\end{aligned}
\end{equation}
\begin{eqnarray}
\text{ s.t. (\ref{eqn:prop_sol3b})-(\ref{eqn:prop_sol3e}) and} \nonumber \\ 
0 \leq x_n^i, x_m^i, y_m^i \leq 1, z_n^i \in \mathbb{R}_0^+, \nonumber \\
\forall n, m \in \mathcal{N}, m \neq n, \forall i \in \mathcal{I}, \label{eqn:prop_sol3f}
\end{eqnarray}
{where $F_{\mathbf{RP}}(\mathbf{x},\mathbf{y},\mathbf{z}) = F(\mathbf{x},\mathbf{y},\mathbf{z})$} in Eq.~(\ref{eqn:prop_sol30a0}).

If all variables $x_n^i$, $x_m^i$, and $y_m^i$ are binary values (i.e., feasible solution is obtained), the algorithm will stop immediately. Otherwise, it will break the ($\mathbf{RP}$) into subproblems ($\mathbf{SP}$s), i.e., branch problems. In this case, the ($\mathbf{SP}$)s fix one of relaxed decision variables (i.e., $x_n^i$, $x_m^i$, or $y_m^i$) to be ``0" at the left branch and ``1" at the right branch. We denote the fixed decision variable to be $\zeta \in \{x_n^i, x_m^i, y_m^i\}$ with $F_{\mathbf{SP}}(\mathbf{x},\mathbf{y},\mathbf{z}) = F_{\mathbf{RP}}(\mathbf{x},\mathbf{y},\mathbf{z})$.

In the BBA, each iteration does not need to search all branch problems (or leaves). Instead, ($\mathbf{SP}$) is pruned if one of these two following conditions is met: (1) $\Psi^c > \beta_U$ or (2) $\xi < \tau$. In this case, $\Psi^c$ and $\beta_U$ refer to the current total average delay and the upper bound of the total average delay. Meanwhile, $\xi$ and $\tau$ represent the integrality between relaxed and rounded decision variables $G$, where $G = \{x_n^i, x_m^i, y_m^i\}$, and integrality gap threshold, respectively. Then, the final solutions ${\hat x}_n^i$, ${\hat x}_m^i$, ${\hat y}_m^i$, and ${\hat z}_n^i$ with $\forall n, m \in \mathcal{N}, m \neq n, \forall i \in \mathcal{I}$ are obtained if the total average delay ${\hat \Psi} = \Psi \leq \Psi^c$ for all $x_n^i$, $x_m^i$, $y_m^i$, and $z_n^i$ in the problem, where $\Psi$ refers to the incumbent total average delay. To guarantee that the near-optimal solution exists, we set an optimality tolerance $\eta$ as a non-negative value. Specifically, ${\hat x}_n^i$, ${\hat x}_m^i$, ${\hat y}_m^i$, and ${\hat z}_n^i$ are $\eta-$optimal when the total average delay ${\hat \Psi}$ is tightened within the bounds and the difference between upper bound $\beta_U$ and lower bound $\beta_L$ is less than the optimality tolerance as expressed below:
\begin{equation}
\label{eqn:conv_an1}
\beta_L \leq {\hat \Psi} \leq \beta_U,
\end{equation}
\begin{equation}
\label{eqn:conv_an2}
	\beta_U - \beta_L \leq \eta.
\end{equation}

While the BBA handles the feasibility of the ($\mathbf{SP}$)s and the optimality of ($\mathbf{P}_3$), the interior-point method (IPM) is used to solve the nonlinear continuous relaxation of the ($\mathbf{SP}$)s. In particular, we derive an interior-point subproblem ($\mathbf{IP}$) from the ($\mathbf{SP}$) as the approximated problem with additional logarithmic barrier functions and non-negative slack variables to eliminate the inequality operators such that:
\begin{equation}
\label{eqn:prop_sol5a0}
\begin{aligned}
&F_{\mathbf{IP_\gamma}}(\mathbf{x},\mathbf{y},\mathbf{z},\boldsymbol{\sigma}_1,\boldsymbol{\sigma}_2,\boldsymbol{\sigma}_3,\boldsymbol{\sigma}_4)  = \\
&F_{\mathbf{SP}}(\mathbf{x},\mathbf{y},\mathbf{z}) - \gamma\bigg[\sum_{j_1=1}^N \text{log } \sigma_1^{j_1}+ \sum_{j_2=1}^{M} \text{log } \sigma_2^{j_2} \\
&+ \sum_{j_3=1}^{J_{\mathbf{x}}+J_{\mathbf{y}}} \text{log } \sigma_3^{j_3} + \sum_{j_4=1}^{J_{\mathbf{x}}+J_{\mathbf{y}}}\text{log } \sigma_4^{j_4}\bigg], \\
\end{aligned}
\end{equation}
and
\begin{equation}
\label{eqn:prop_sol5a}
\begin{aligned}
(\mathbf{IP}) \phantom{5} &\min_{\{\substack{\mathbf{x},\mathbf{y},\mathbf{z},\\ \boldsymbol{\sigma}_1,\boldsymbol{\sigma}_2,\boldsymbol{\sigma}_3,\boldsymbol{\sigma}_4}\}}  F_{\mathbf{IP_\gamma}}(\mathbf{x},\mathbf{y},\mathbf{z},\boldsymbol{\sigma}_1,\boldsymbol{\sigma}_2,\boldsymbol{\sigma}_3,\boldsymbol{\sigma}_4), 
\end{aligned}
\end{equation}
\begin{eqnarray}
\text{s.t.} \quad
\Theta \big(x_n^i\big) + \sigma_1^{j_1} = 0, \forall n,j_1 \in \mathcal{N}, \label{eqn:prop_sol5b} \\
 Q_n^i \big(x_m^i\big) - V y_m^i - z_n^i + \sigma_2^{j_2} = 0, \forall j_2 \in [1, M], \nonumber \\
\forall n, m \in \mathcal{N}, m \neq n, \forall i \in \mathcal{I}, \label{eqn:prop_sol5c} \\
\Gamma \big(y_m^i\big) = 0, \forall n \in \mathcal{N}, \forall i \in \mathcal{I},\label{eqn:prop_sol5d} \\
G - \sigma_3^{j_3} = 0, G + \sigma_4^{j_4} = 1,  \nonumber \\
\forall j_3, j_4 \in [1, J_{\mathbf{x}}+J_{\mathbf{y}}], \nonumber \\
\forall n, m \in \mathcal{N}, m \neq n, \forall i \in \mathcal{I}, \label{eqn:prop_sol5e} \\
z_n^i \in \mathbb{R}_0^+, \forall n \in \mathcal{N}, \forall i \in \mathcal{I}, \label{eqn:prop_sol5f}
\end{eqnarray}
where $M \defeq {\underset{i=1}{\overset{I}{\sum}}}{\underset{n=1}{\overset{N}{\sum}}}M_n^i$, $\gamma > 0$ is the barrier parameter \cite{Benson:2011}, $\boldsymbol{\sigma}_1, \boldsymbol{\sigma}_2 \in \mathbb{R}_0^+$ are slack variables in the inequality constraints (\ref{eqn:prop_sol3b}) and (\ref{eqn:prop_sol3d}), respectively, and $\boldsymbol{\sigma}_3, \boldsymbol{\sigma}_4 \in \mathbb{R}_0^+$ are slack variables for lower and upper bounds of $G$ in the inequality constraints (\ref{eqn:prop_sol3f}), respectively. In this case, we define $\boldsymbol{\sigma}_1 \defeq [\sigma_1^1,\ldots,\sigma_1^{j_1},\ldots,\sigma_1^{N}]$, $\boldsymbol{\sigma}_2 \defeq [\sigma_2^1,\ldots,\sigma_2^{j_2},\ldots,\sigma_2^{M}]$, $\boldsymbol{\sigma}_3 \defeq [\sigma_3^1,\ldots,\sigma_3^{j_3},\ldots,\sigma_3^{J_{\mathbf{x}}+J_{\mathbf{y}}}]$, and $\boldsymbol{\sigma}_4 \defeq [\sigma_4^1,\ldots,\sigma_4^{j_4},\ldots,\sigma_4^{J_{\mathbf{x}}+J_{\mathbf{y}}}]$. To make $F_{\mathbf{IP_\gamma}}(\mathbf{x},\mathbf{y},\mathbf{z},\boldsymbol{\sigma}_1,\boldsymbol{\sigma}_2,\boldsymbol{\sigma}_3,\boldsymbol{\sigma}_4)$ equal to the minimum $ F_{\mathbf{SP}}(\mathbf{x},\mathbf{y},\mathbf{z})$, we need to update $\gamma$ in decreasing order for each iteration such that it converges to zero \cite{Waltz:2005}. To solve the approximated problem $F_{\mathbf{IP_\gamma}}(\mathbf{x},\mathbf{y},\mathbf{z},\boldsymbol{\sigma}_1,\boldsymbol{\sigma}_2,\boldsymbol{\sigma}_3,\boldsymbol{\sigma}_4)$, a conjugate gradient technique \cite{Steihaug:1983} is adopted to minimize a quadratic approximation of the approximated problem at each step. Let $\mathcal{L}(\mathbf{p},\boldsymbol{\sigma},\boldsymbol{\lambda})$ denote the Lagrangian expression containing the objective function in Eq.~(\ref{eqn:prop_sol5a}) along with the corresponding constraints~(\ref{eqn:prop_sol5b})-(\ref{eqn:prop_sol5e}) and a vector of non-negative Lagrangian multipliers $\boldsymbol{\lambda}$, where $\boldsymbol{\lambda}$ contains $\lambda_1^{j_1}, \forall j_1 \in \mathcal{N}, \lambda_2^{j_2}, \forall j_2 \in [1,M], \lambda_3^{j_3}, \forall j_3 \in [1, J_{\mathbf{x}}+J_{\mathbf{y}}], \lambda_4^{j_4}, \forall j_4 \in [1, J_{\mathbf{x}}+J_{\mathbf{y}}]$, and $\lambda_5^{j_5}, \forall j_5 \in [1, J_{\mathbf{z}}]$.
Then, $\mathcal{L}(\mathbf{p},\boldsymbol{\sigma},\boldsymbol{\lambda})$ can be expressed by
\begin{equation}
\label{eqn:prop_sol6}
\begin{aligned}
&\mathcal{L}(\mathbf{p},\boldsymbol{\sigma},\boldsymbol{\lambda}) = F_{\mathbf{SP}}(\mathbf{p}) - \gamma\bigg[\sum_{j_1=1}^N \text{log } \sigma_1^{j_1}+ \sum_{j_2=1}^{M} \text{log } \sigma_2^{j_2} \\
&+ \sum_{j_3=1}^{J_{\mathbf{x}}+J_{\mathbf{y}}} \text{log } \sigma_3^{j_3} + \sum_{j_4=1}^{J_{\mathbf{x}}+J_{\mathbf{y}}}\text{log } \sigma_4^{j_4}\bigg] + \sum_{j_1=1}^N\lambda_1^{j_1}\bigg(\Theta \big(x_n^i\big) + \sigma_1^{j_1}\bigg)\\
&+ \sum_{j_2=1}^{M} \lambda_2^{j_2}\bigg(Q_n^i \big(x_m^i\big) - V y_m^i - z_n^i + \sigma_2^{j_2}\bigg) \\
&+\sum_{j_3=1}^{J_{\mathbf{x}}+J_{\mathbf{y}}}\lambda_3^{j_3}\bigg(G - \sigma_3^{j_3}\bigg) + \sum_{j_4=1}^{J_{\mathbf{x}}+J_{\mathbf{y}}}\lambda_4^{j_4}\bigg(G + \sigma_4^{j_4} - 1\bigg) \\
&+ \sum_{j_5=1}^{J_{\mathbf{z}}}\lambda_5^{j_5} \Gamma \big(y_m^i\big),
\end{aligned}
\end{equation}
where $\mathbf{p} = (\mathbf{x},\mathbf{y},\mathbf{z})$ and $\boldsymbol{\sigma} = (\boldsymbol{\sigma}_1,\boldsymbol{\sigma}_2,\boldsymbol{\sigma}_3,\boldsymbol{\sigma}_4)$. We first need to obtain the Lagrangian multipliers by solving the following equation before using the conjugate gradient:
\begin{equation}
\label{eqn:prop_sol7}
\begin{aligned}
\nabla_\mathbf{p} \mathcal{L}(\mathbf{p},\boldsymbol{\sigma},\boldsymbol{\lambda}) &= \nabla_\mathbf{p} F_{\mathbf{SP}}(\mathbf{p}) + \sum_{j_1=1}^N\lambda_1^{j_1} \nabla_\mathbf{p} \Theta \big(x_n^i\big) \\
&+ \sum_{j_2=1}^{M} \lambda_2^{j_2} \nabla_\mathbf{p} \bigg(Q_n^i \big(x_m^i\big) - V y_m^i - z_n^i\bigg) \\
&+\sum_{j_3=1}^{J_{\mathbf{x}}+J_{\mathbf{y}}}\lambda_3^{j_3} \nabla_\mathbf{p}G  + \sum_{j_4=1}^{J_{\mathbf{x}}+J_{\mathbf{y}}}\lambda_4^{j_4} \nabla_\mathbf{p}G  \\
&+ \sum_{j_5=1}^{J_{\mathbf{z}}}\lambda_5^{j_5} \nabla_\mathbf{p} \Gamma \big(y_m^i\big) = 0.
\end{aligned}
\end{equation}
Then, we can use a step size ($\Delta{\mathbf{p}}, \Delta{\boldsymbol{\sigma}}$) iteratively to minimize the quadratic approximation such that:
\begin{equation}
\label{eqn:prop_sol8a}
\begin{aligned}
\min_{\{\Delta{\mathbf{p}}, \Delta{\boldsymbol{\sigma}}\}}  F(\Delta{\mathbf{p}}, \Delta{\boldsymbol{\sigma}}),
\end{aligned}
\end{equation}
\begin{eqnarray}
\text{s.t.} \quad
\Theta \big(x_n^i\big) + \mathbf{W}_1\Delta{\mathbf{p}} + \Delta{\boldsymbol{\sigma}} = 0, \forall n \in \mathcal{N},\label{eqn:prop_sol8b} \\
Q_n^i \big(x_m^i\big) - V y_m^i - z_n^i + \mathbf{W}_2\Delta{\mathbf{p}} +  \Delta{\boldsymbol{\sigma}} = 0, \nonumber \\ 
\forall n, m \in \mathcal{N}, m \neq n, \forall i \in \mathcal{I}, \label{eqn:prop_sol8c} \\
G + \mathbf{W}_3\Delta{\mathbf{p}} + \Delta{\boldsymbol{\sigma}} = 0, G - 1 + \mathbf{W}_4\Delta{\mathbf{p}} + \Delta{\boldsymbol{\sigma}}= 0,  \nonumber \\
\forall n, m \in \mathcal{N}, m \neq n, \forall i \in \mathcal{I}, \label{eqn:prop_sol8e} \\
\Gamma \big(y_m^i\big) + \mathbf{W}_5\Delta{\mathbf{p}} = 0, \forall n \in \mathcal{N}, \forall i \in \mathcal{I}, \label{eqn:prop_sol8d} \\
\end{eqnarray}
where 
\begin{equation}
\label{eqn:prop_sol8ab}
\begin{aligned}
F(\Delta{\mathbf{p}}, \Delta{\boldsymbol{\sigma}}) &=\nabla_\mathbf{p} F_{\mathbf{SP}}(\mathbf{p})^\textrm{T}\Delta{\mathbf{p}} \\
&+ \frac{1}{2}\Delta{\mathbf{p}}^\textrm{T}\nabla_\mathbf{p}^2 \mathcal{L}(\mathbf{p},\boldsymbol{\sigma},\boldsymbol{\lambda}) \Delta{\mathbf{p}} \\
&+ \gamma \mathbf{e}^\textrm{T}\mathbf{S}^{-1}\Delta{\boldsymbol{\sigma}} + \frac{1}{2}\Delta{\boldsymbol{\sigma}}^\textrm{T}\mathbf{S}^{-1}\mathbf{A}\Delta{\boldsymbol{\sigma}}, 
\end{aligned}
\end{equation}
$\mathbf{e}$ is a vector of ones corresponding to sizes of $\Theta \big(x_n^i\big)$, $\big(Q_n^i \big(x_m^i\big) - V y_m^i - z_n^i\big)$, and $G$, $\mathbf{S}$ and $\mathbf{A}$ are the diagonal matrices with the elements of $\boldsymbol{\sigma}$ and $\boldsymbol{\lambda}$ on the diagonal, respectively. Furthermore, $\mathbf{W} = (\mathbf{W}_1, \mathbf{W}_2, \mathbf{W}_3, \mathbf{W}_4,\mathbf{W}_5)$ represents the Jacobian of the constraint functions $\Theta \big(x_n^i\big)$, $\big(Q_n^i \big(x_m^i\big) - V y_m^i - z_n^i\big)$, $G$, and $\Gamma \big(y_m^i\big)$.

\begin{algorithm}[]
\caption{Centralized Solution with iBBA-IPM} \label{MINLP-iBBA-IPM}

\begin{algorithmic}[1]

\STATE $\mathcal{R}_a$: the set of $r_a$ active problems, $r_c$: the current problem

\STATE Set $G = \{x_n^i, x_m^i, y_m^i \} \in [0,1] \text{, } z_n^i \in \mathbb{R}_0^+\text{, } \forall n,m \in \mathcal{N} \text{, } m \neq n\text{, }\forall i \in \mathcal{I}$

\STATE $r_c \leftarrow r_a \leftarrow 1 $ /* Set root problem */

\STATE $\beta_L \leftarrow -\infty$, $\beta_U \leftarrow +\infty$, $\Psi \leftarrow \beta_U$ 

\STATE Set $\tau$ and $\eta$ /* Integrality gap \& optimality tolerance */

\STATE $\text{Solve }(\mathbf{RP})$ for $r_c$

\IF {$\forall x_n^i, x_m^i, y_m^i \in \{0,1\}$}

\STATE Store ${\hat x}_n^i, {\hat x}_m^i, {\hat y}_m^i, {\hat z}_n^i$, and ${\hat \Psi}$, $\forall n,m \in \mathcal{N} \text{, }\forall i \in \mathcal{I}$

\RETURN /* Prune all active problems */
	
\ENDIF

\WHILE{$\mathcal{R}_a \neq \varnothing \text{ and } \beta_U - \beta_L > \eta$}

\STATE $r_c \leftarrow r_a$ /* an ($\mathbf{SP}$) based on the depth-first search */

\STATE $r_a \leftarrow r_a - 1$  /* Remove the ($\mathbf{SP}$) from the set */

\STATE $\text{Solve subproblem }(\mathbf{IP})$ for $r_c$

\IF{$\text{($\mathbf{IP})$ for $r_c$ is infeasible}$}

\STATE Prune $r_c$, {\bf exit}

\ELSE

\STATE Set current $G$ and $z_n^i$, $\forall n,m \in \mathcal{N} \text{, }\forall i \in \mathcal{I}$

\STATE Calculate $\xi \leftarrow |G - \text{round}(G)|$ 


\IF{$\xi < \tau$}

\IF{$\Psi^c < \Psi$}

\STATE Store current $G$ and $z_n^i$, $\forall n,m \in \mathcal{N} \text{, }\forall i \in \mathcal{I}$

\STATE Set $\Psi \leftarrow \Psi^c$ and $\beta_U \leftarrow \Psi^c$

\ENDIF

\STATE Prune $r_c$, {\bf exit}

\ELSE

\STATE Choose $\zeta \in G$ 

\STATE $r_a \leftarrow r_a + 2$  /* Add 2 ($\mathbf{SP})$s to the set */

\STATE Update  $\beta_L \leftarrow \Psi^c$ 

\STATE Update the constraints with $\zeta \leftarrow 0$, $\zeta \leftarrow 1$, {\bf exit}

\ENDIF

\ENDIF

	
\STATE Store ${\hat x}_n^i, {\hat x}_m^i, {\hat y}_m^i, {\hat z}_n^i$, and ${\hat \Psi}$, $\forall n,m \in \mathcal{N} \text{, }\forall i \in \mathcal{I}$



\ENDWHILE

\end{algorithmic}
\end{algorithm}

\subsection{Algorithm Complexity and Convergence Analysis of Improved BBA-IPM}
\label{sec:BBIPM-AC}

Algorithm \ref{MINLP-iBBA-IPM} describes the pseudocode of the improved BBA-IPM (iBBA-IPM) solution. First, we evaluate the complexity of the algorithm. The complexity of Algorithm \ref{MINLP-iBBA-IPM} is the lowest when the final solutions ${\hat x}_n^i, {\hat x}_m^i, {\hat y}_m^i, {\hat z}_n^i$ with $\forall n, m \in \mathcal{N}, m \neq n, \forall i \in \mathcal{I}$ and the final total average delay ${\hat \Psi}$ are found at ($\mathbf{RP}$). However, the complexity becomes higher when ($\mathbf{SP}$)s are active at left and right branches of the iBBA-IPM's tree. In particular, we apply a \emph{depth-first search} method by selecting the latest created ($\mathbf{SP}$) as the next ($\mathbf{SP}$)~\cite{Smith:1979} to find the $t$-th total average delay, i.e., $\Psi_t$. 
Suppose that $T$ is the expected number of total average delays $\Psi_t, \forall t \in [1,T]$, searched in the iBBA-IPM's tree, given that the first solution is found at depth $\varrho$. This $T$ is bounded above and below by linear functions of $\varrho$ to achieve polynomial complexity. This is formally described in Theorem~\ref{theorem1b}.
\begin{theorem}\label{theorem1b} 
Given that the first solution is found at depth $\varrho$ in the iBBA-IPM's tree, the expected number of total average delays searched in the iBBA-IPM's tree, i.e., $T$, follows polynomial complexity at the depth $\varrho$ when $T$ is bounded above and below by linear functions of $\varrho$.
\end{theorem}
\begin{proof}
See Appendix~\ref{appx:theorem2}.
\end{proof}

Next, we show that the iBBA-IPM algorithm converges after a finite number of steps under the optimality tolerance $\eta > 0$ of the optimization problem in ($\mathbf{P}_2$). This is formally stated in Theorem \ref{theorem3}. 
\begin{theorem}\label{theorem3}
The iBBA-IPM algorithm converges after a finite number of steps under the optimality tolerance $\eta > 0$. Specifically, there exists step $\theta_{\eta} \in \mathbb{N}$ for any $\eta > 0$ such that:
\begin{equation}
\label{eqn:prop_sol12}
\begin{aligned}
\beta_U^{(\theta_\eta)} - \beta_L^{(\theta_\eta)} \leq \eta,
\end{aligned}
\end{equation}
and $\beta_U^{(\theta_\eta)}$ is within the optimality tolerance of the $F(\mathbf{x},\mathbf{y},\mathbf{z})$.
\end{theorem}
\begin{proof}
See Appendix~\ref{appx:theorem4}.
\end{proof}

\section{Distributed Cooperative Caching{-Delivering} Solution}
\label{sec:DS}

Despite the fact that the aforementioned proposed solution (i.e., the centralized solution) can find the final solution within 1\% to minimize the total average delay for the whole network, it faces two issues. First, the centralized solution requires a centralized computer with the whole network topology and information to be able to solve the nested dual optimization problem. Second, it incurs significant communication overheads among MENs, especially for a large number of MENs. In this section, we introduce a distributed suboptimal solution, which can address these issues. Specifically, under the distributed solution, each MEN first finds the locally optimal caching policy based on the local users' demands and its storage capability. Then, each MEN will communicate with its directly connected nodes to find cooperations through discovering {the duplicate contents}. 

To obtain locally optimal caching policy under the distributed solution, each MEN-$n$ needs to consider three cases when its local user $u$ requests content $i$:
\begin{itemize}
	\item \emph{Case 1}: If MEN-$n$ has the requested content, the delay to deliver the content will be $x^i_n d_\alpha$, where $d_\alpha = \frac{c_i}{b_n^u}$.
	\item \emph{Case 2}: If MEN-$n$ does not have the requested content but the BS (MEN-$N$) has the content, the MEN-$n$ will download the content from the BS, and thus the delay will be $(1 - x^i_n)x^i_Nd_{\beta}^\ddagger$, where $d_{\beta}^\ddagger \defeq \Big(d_\alpha + \frac{c_i}{l^N_n}\Big)$.
	\item \emph{Case 3}: If MEN-$n$ and the BS do not have the requested content, the content will downloaded from the CS via the BS, and thus the delay will be $(1 - x^i_n)(1 - x^i_N)d_{\delta}$, where $d_{\delta} = \Big(d_\alpha + \frac{c_i}{l^N_n} + \frac{c_i}{B}\Big)$.
\end{itemize}
Note that \emph{Case 2} is ignored when users connected directly to the BS. Moreover, since the bandwidth between the BS and CS is usually much higher than that among MENs \cite{Poularakis2:2016}, for \emph{Case 3}, the delay to download the content when users are connected to the BS directly becomes $(1 - x^i_N)d_{\delta}^\ddagger$, where $d_{\delta}^\ddagger \defeq \Big(d_\alpha + \frac{c_i}{B}\Big)$. Based on the aforementioned analysis, we then formulate the locally optimal caching optimization problem, i.e., optimal caching at each node without horizontal cooperation, for MEN-$n$ (where $n \in [1, N-1]$) as ($\mathbf{P}_4$) with the following equation:
\begin{equation}
\label{eqn:prop_sol21a}
\begin{aligned}
(\mathbf{P}_4) \phantom{5} \min_{\mathbf{x}} F_{n}(\mathbf{x}),
\end{aligned}
\end{equation}
\begin{eqnarray}
\text{s.t }\sum_{i=1}^Ix_n^ic_i \leq s_n, n \in [1, N-1],  \label{eqn:prop_sol21b} \\
x_n^i, x_N^i \in \{0, 1\},  n \in [1, N-1], \forall i \in \mathcal{I}, \label{eqn:prop_sol21c}
\end{eqnarray}
and for the BS as ($\mathbf{P}_5$) with the following expression:
\begin{equation}
\label{eqn:prop_sol21d}
\begin{aligned}
(\mathbf{P}_5) \phantom{5} & \min_{\mathbf{x}} F_{N}(\mathbf{x}),
\end{aligned}
\end{equation}
\begin{eqnarray}
\text{s.t }\sum_{i=1}^Ix_N^ic_i \leq s_N, \label{eqn:prop_sol21e} \\
x_N^i \in \{0, 1\}, \forall i \in \mathcal{I}, \label{eqn:prop_sol21f}
\end{eqnarray}
{where}
\begin{equation}
\label{eqn:prop_sol21a0}
\begin{aligned}
F_{n}(\mathbf{x}) &= \sum_{i=1}^I f_n^i \Bigg[ \sum_{u=1}^{U_n} \bigg( x^i_n d_\alpha + \\
&(1 - x^i_n)x^i_Nd_{\beta}^\ddagger + (1 - x^i_n)(1 - x^i_N)d_{\delta} \bigg)\Bigg] \\
&= \sum_{i=1}^I f_n^i \Bigg[ \sum_{u=1}^{U_n} \bigg( x^i_n d_\alpha + \\
&(1 - x^i_n)\Big(x^i_Nd_{\beta}^\ddagger + (1 - x^i_N)d_{\delta}\Big) \bigg)\Bigg], \\
\end{aligned}
\end{equation}
{and}
\begin{equation}
\label{eqn:prop_sol21d0}
\begin{aligned}
F_{N}(\mathbf{x}) =\sum_{i=1}^I f_N^i \Bigg[ \sum_{u=1}^{U_N} \bigg( x^i_N d_\alpha + (1 - x^i_N)d_{\delta}^\ddagger \bigg)\Bigg].
\end{aligned}
\end{equation}
It is shown that the problems ($\mathbf{P}_4$) and ($\mathbf{P}_5$) are standard binary linear programming {which are generally NP-complete~\cite{Karp:1972}. Nonetheless, the optimization occurs at each MEN only with a much lower number of variables (than the centralized problem). In this way, those problems then can be solved effectively using popular solvers}. 

After obtaining the locally optimal caching decisions at each MEN, the cooperations among MENs to find {the duplicate contents} is then carried out. In particular, each MEN-$n$ will consider the {duplication} of its current cached contents with current ones of its directly connected MEN-$m$, where $m \in \mathcal{M}_n^i$, at a particular time. The flowchart in Fig.~\ref{fig:Flowchart_DS} shows how the cooperations among MENs work in this solution. The key idea of this cooperation process is that we need to find the minimum total average delay $\Psi$ in the network by minimizing the duplicate contents among MENs. The process terminates when all candidate contents have been checked at each MEN-$n$. The complete pseudocode of the distributed solution is shown in Algorithm \ref{MINLP-DSA}. 

\begin{figure}[!t]
	\centering
	\includegraphics[scale=0.53]{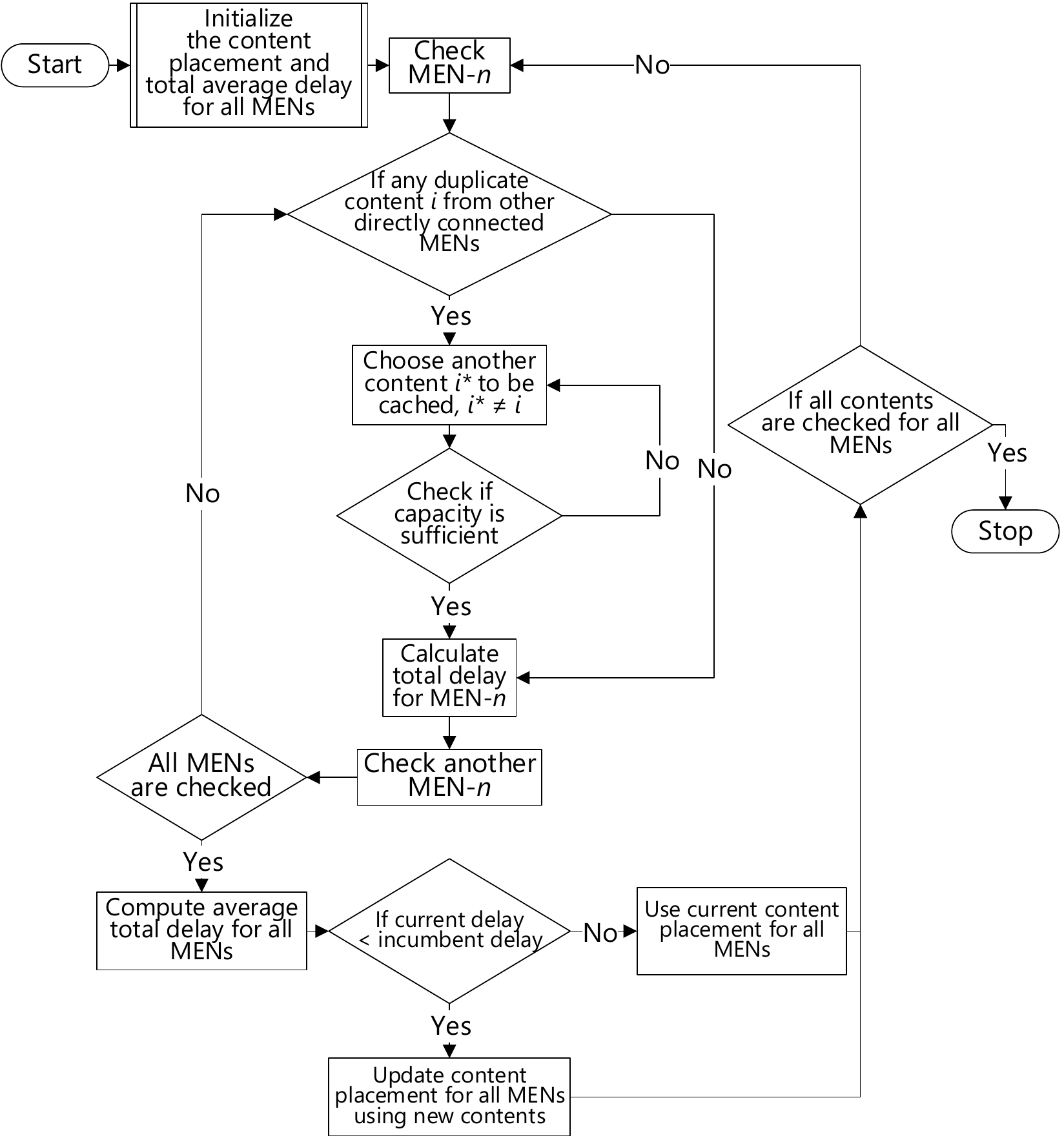}
	\caption{Flowchart of cooperations scheme among MENs for the distributed solution.}
	\label{fig:Flowchart_DS}
\end{figure}

\begin{algorithm}[]
\caption{Distributed Solution} \label{MINLP-DSA}

\begin{algorithmic}[1]

\STATE Initialize $x_n^i \in {0,1}, \forall n \in \mathcal{N} \text{, }\forall i \in \mathcal{I}$ from locally optimal policy problems ($\mathbf{P}_4$) and ($\mathbf{P}_5$)

\STATE Check $s^*_n, \forall n \in \mathcal{N}$ /* Current capacity of each MEN */

\STATE Calculate $\Psi$ /* total average delay of $F_n(\mathbf{x})$ in problem ($\mathbf{P}_4$) and $F_N(\mathbf{x})$ in problem ($\mathbf{P}_5$) */


\FOR{$\forall i \in \mathcal{I}$}

\FOR{$\forall n \in \mathcal{N}$}

\IF{$x_n^i = x_m^i = 1, \forall m \in \mathcal{M}_n^i \subset \mathcal{N}, m \neq n$}

\STATE Set $s^*_n \leftarrow s^*_n - c_i$ 

\STATE Select $i^*, i^* \neq i,  \forall i^* \in \mathcal{I}$ 

\IF{$s^*_n + c_{i^*} \leq s_n$}

\STATE Set $x_n^{i^*} \leftarrow 1$ /* Store new candidate content */ 

\STATE Obtain information from MEN-$m, \forall m \in \mathcal{M}_n^i$

\STATE Calculate $\Psi^c_n$ /* Current total average delay */

\ENDIF

\ENDIF

\ENDFOR

\STATE Compute the average of all $\Psi^c_n$'s, $\forall n \in \mathcal{N}$ into $\Psi^c$ 

\IF{$\Psi^c < \Psi$}

\STATE  Update $x_n^i, \forall n \in \mathcal{N} \text{, }\forall i \in \mathcal{I}$

\STATE  $\Psi \leftarrow \Psi^c$

\ENDIF

\ENDFOR

\STATE Store final ${\hat x}_n^i, \forall n \in \mathcal{N} \text{, }\forall i \in \mathcal{I}$

\STATE Calculate final ${\hat \Psi}$ using the solution in Line 22


\end{algorithmic}
\end{algorithm}

The distributed solution algorithm has polynomial complexity. This is shown by the double looping when each content $i$ is checked at each MEN-$n$. In addition, if an MEN-$n$ has duplicate contents with at least one directly connected MEN-$m$, the MEN-$n$ will check the rest of the contents (i.e., $I^* = I - 1$). Obviously, the algorithm has polynomial time complexity $O(I \times N \times I^*)$.
\section{Illustrative Case Study}
\label{sec:IE}

\begin{figure}[!]
	\centering
	\includegraphics[scale=0.35]{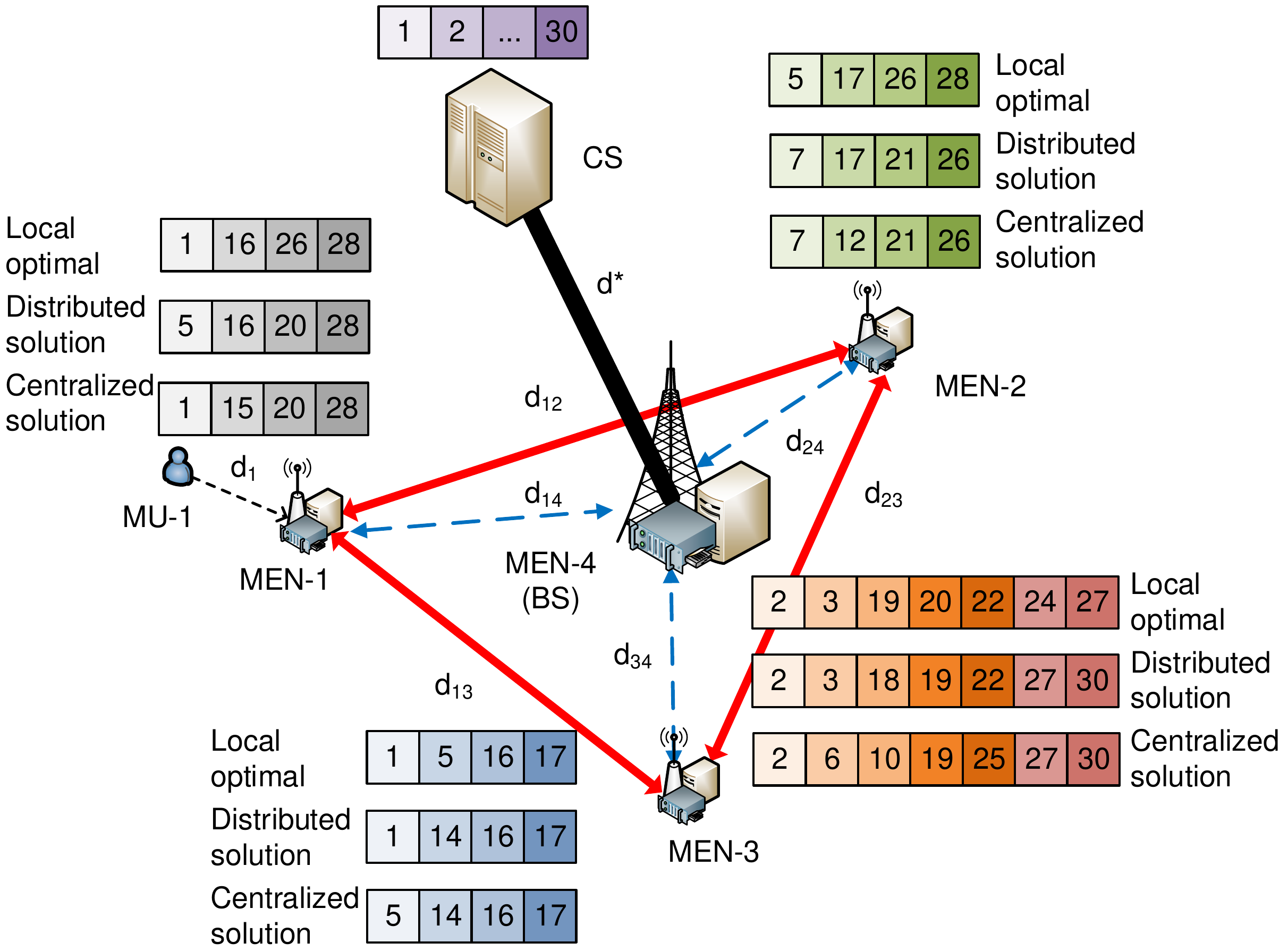}
	\caption{An illustrative example of caching policies obtained by locally optimal policy, distributed solution, and centralized solution.}
	\label{fig:illus_example}
\end{figure}

To show the efficiency of the proposed solutions using direct horizontal cooperations among MENs, we present an illustrative example in Fig.~\ref{fig:illus_example}. We use 4 nodes (i.e., 3 MENs and the BS) and 30 available contents with various frequency-of-accesses. The content size is uniformly distributed between 50MB and 200MB. We set the storage capacity at 600MB for MEN-1 to MEN-3 and 1GB for the BS. Suppose that the bandwidth between MEN-$n$ and MEN-$m$ ($\forall n,m \in \mathcal{N},$ and $n \neq m$) is equal. Based on the aforementioned parameter settings, each MEN and the BS can cache up to 4-7 contents. 

As can be seen in Fig.~\ref{fig:illus_example}, we show caching policies obtained by the locally optimal policy and the proposed solutions (i.e., distributed and centralized solutions). Thanks to horizontal cooperations among MENs in distributed and centralized solutions, MENs can share cached contents to minimize the total average delay for the whole network instead of minimizing the delay for each individual MEN (as in the locally optimal policy). For example, under the centralized solution, there is no duplicate content cached at all MENs (i.e, each MEN has different set of contents), and thus the number of contents cached at MENs in the network can be maximized (i.e., 19 of 30 available contents are cached in the network). In this way, the centralized solution can find the  near-optimal caching and delivering policy based on the full network topology and information of all nodes in the whole network. 

For the distributed solution, few duplicate contents are cached at MENs. Specifically, content 16 is cached at MEN-1 and MEN-3, while content 17 is stored at MEN-2 and MEN-3. The reason is that the distributed solution cannot achieve optimal solution even though the {duplicate} contents at all MENs {are} minimized. However, compared with locally optimal policy, the distributed solution has much better disparate content distribution due to the horizontal cooperation (i.e., 17 of 30 available contents are cached in the network). When we use the locally optimal policy, more duplicate contents are produced (i.e., only 13 of 30 available contents are cached in the network) due to the locally caching optimization at each MEN (without horizontal cooperation among MENs).

\section{Simulation Results}
\label{sec:NR}

We perform simulations to evaluate the performance of the proposed solutions, i.e., centralized and distributed solutions, with those of other caching policies including greedy policy \cite{Borst2010Distributed,Zhang:2018,Wang:2016}, most FoA policy \cite{Yang:2016,Cui:2017,Hassan:2016,Chen:2017}, locally optimal policy \cite{Hoang:2018,Liu:2017}, and guaranteed greedy policy \cite{Dehghan2015On,Golrezaei:2012}. For the greedy policy, contents are cached as many as possible. For the most FoA policy, contents with high FoA are prioritized to be cached. For the locally optimal policy, the MENs minimize their own delays without considering the cooperation. For the guaranteed greedy policy, contents which maximize caching gain are greedily added into caching storage under guaranteed ($1 - 1/\emph{e}$) factor of the optimal solution. In all simulations, the content size is randomly generated using a uniform distribution between 100 and 300MB, while the FoAs of all MENs follow a Zipf distribution based on the ranks of the contents~\cite{Poularakis:2016,Ao:2015,Poularakis:2014} with shape parameter set at 0.1. The bandwidth between an MU and its associated MEN, and between two directly connected MENs are set at 10 and 45Mbps, respectively. Furthermore, bandwidth between an MEN and the BS is 10Mbps and between the BS and the CS is 60Mbps. Note that the bandwidth between two MENs is usually higher than that between an MEN and the BS because in practice two directly connected MENs are often placed in the same area or close to each other where wired or fast wireless connections can be used~\cite{Poularakis:2016,Ao:2015,Poularakis:2014}. 

\subsection{Evaluation of Proposed vs. Optimal Solutions}

\begin{figure}[!]
	\centering
	\includegraphics[scale=0.5]{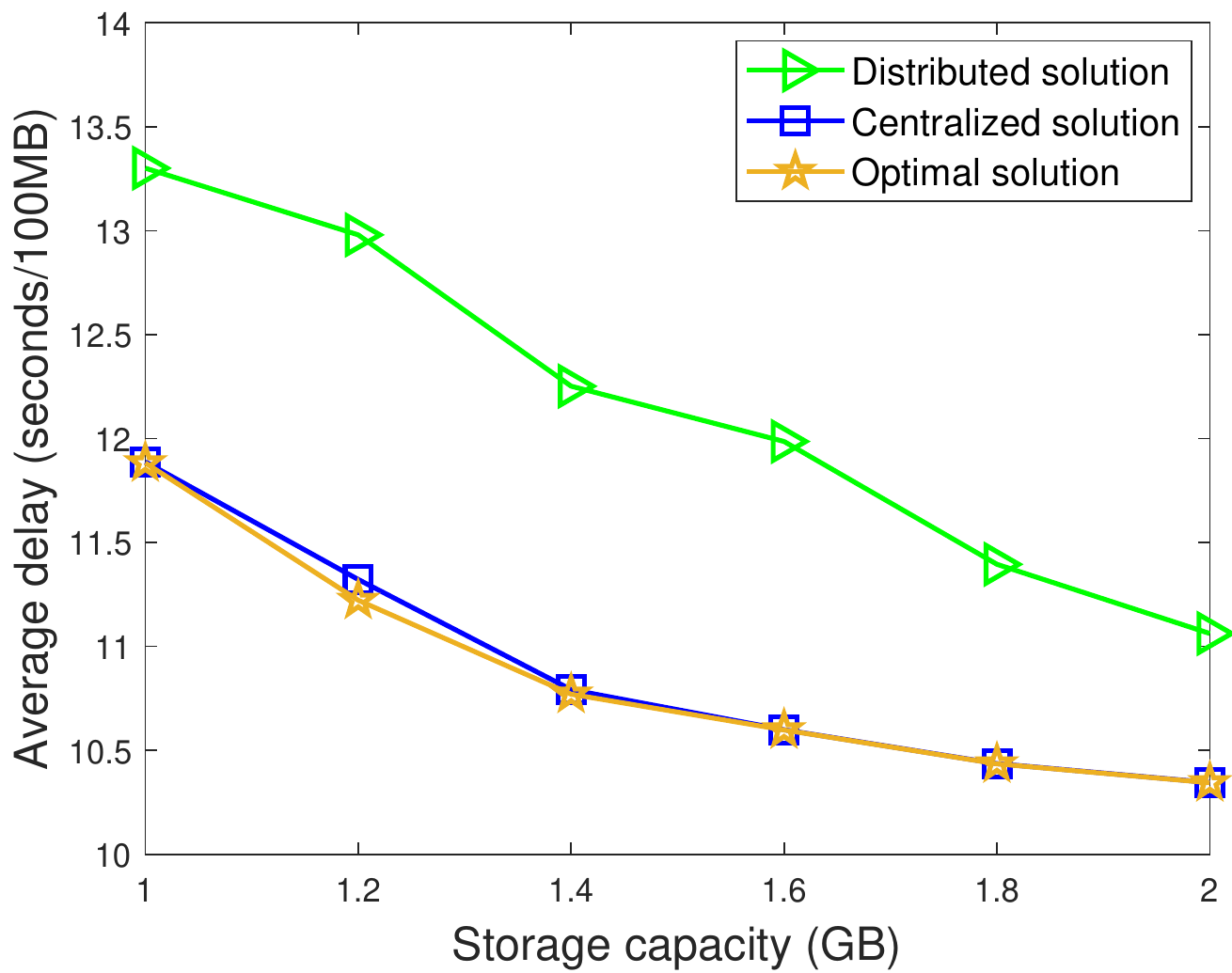}
	\caption{Evaluation of proposed solutions against optimal solution.}
	\label{fig:Avg_delay_opt}
\end{figure}

Fig.~\ref{fig:Avg_delay_opt} demonstrates the total average delay performance between our proposed solutions, i.e., centralized and distributed solutions, and optimal solution when the storage capacity increases from 1 to 2GB and 15 contents are available for 2 MENs. Here, we only consider small-size problems due to the exponential complexity of finding the optimal solution using the exhaustive searching scheme. It can be seen that although the proposed centralized solution cannot guarantee the optimal solution for the original problem, its performance can achieve very close (less than 1\% gap) to that of the optimal solution obtained by the exhaustive searching scheme. Furthermore, the proposed distributed solution still can provide reasonable total average delay gap within 13\% to the optimal solution and the gap becomes smaller as the storage capacity increases.

\subsection{Total Average Delay}  
\label{subsec:total_delay}

In this section, we evaluate the total average delay under the influence of various storage capacities, number of contents, and number of MENs in the network. In particular, we use 0 to 10GB caching capacity, 50 to 200 number of available contents, and 2 to 10 number of MENs. Furthermore, the content size is uniformly distributed between 100MB and 300MB. 

\subsubsection{Effects of Storage Capacity}

\begin{figure}[!]
	\centering
	\includegraphics[scale=0.5]{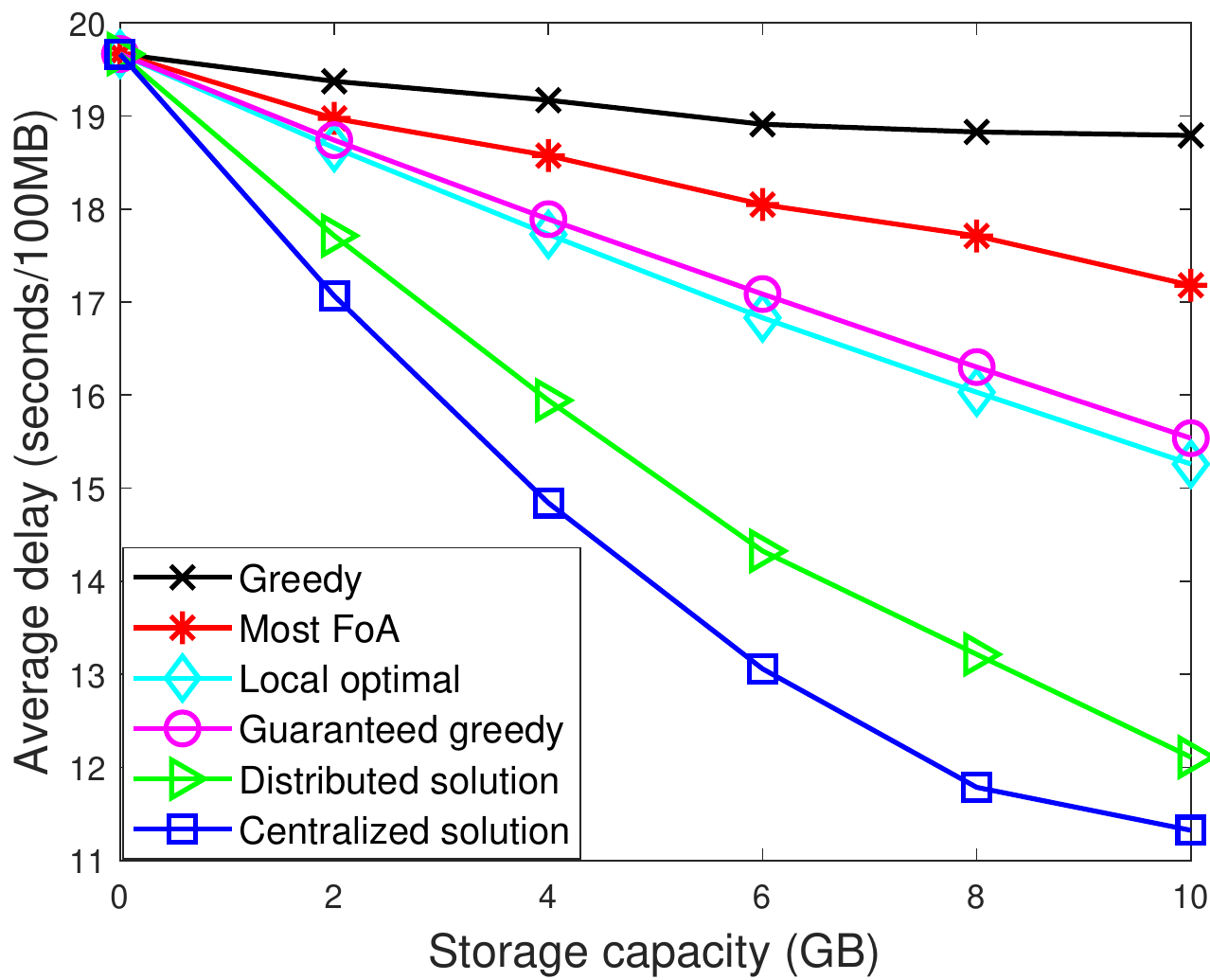}
	\caption{Average delay vs storage capacity.}
	\label{fig:Avg_delay_inc_cap}
\end{figure}

\begin{figure*}[!]
	\begin{center}
		$\begin{array}{ccc} 
		\epsfxsize=2.35 in \epsffile{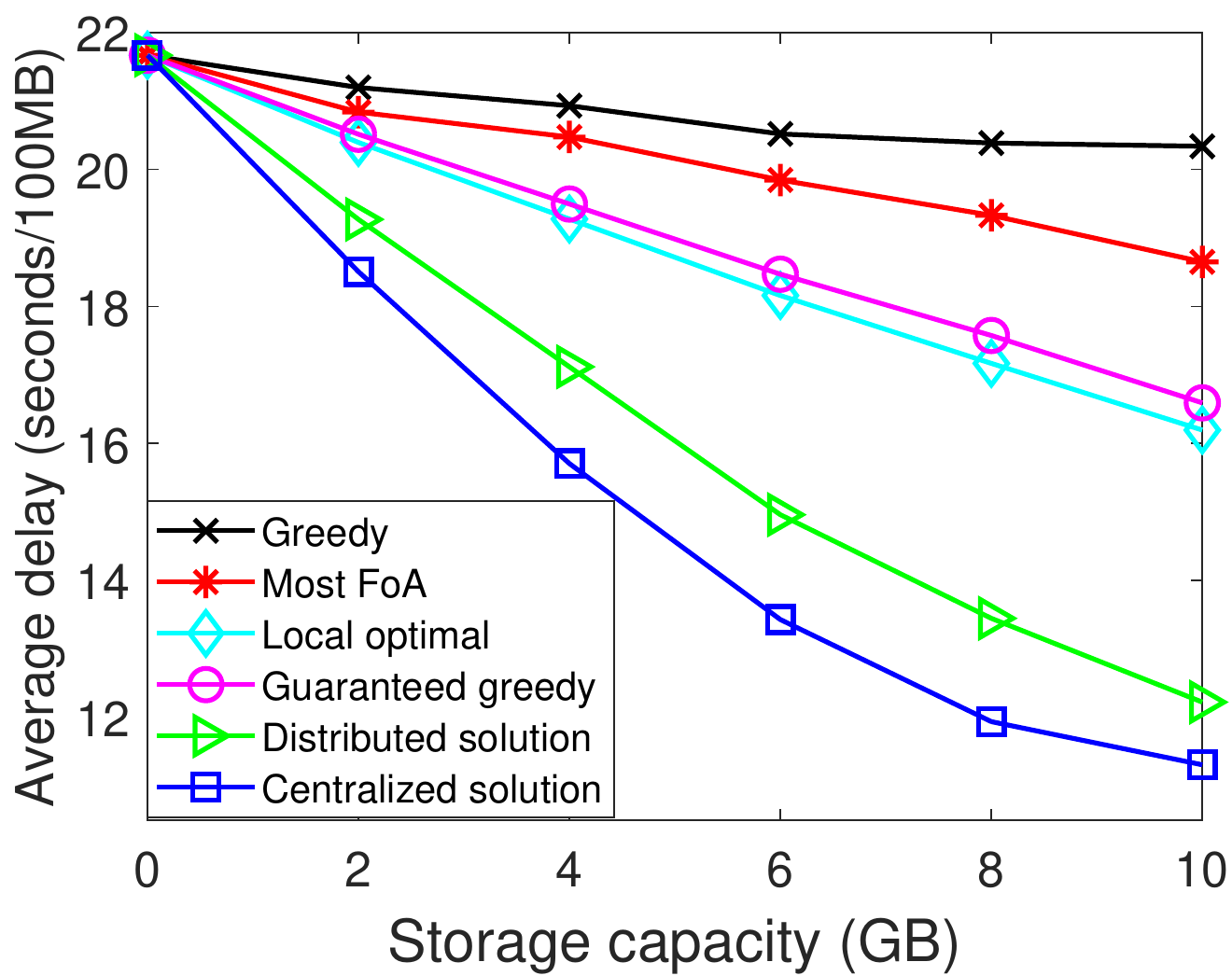} &
		\hspace*{-.37cm}
		\epsfxsize=2.35 in \epsffile{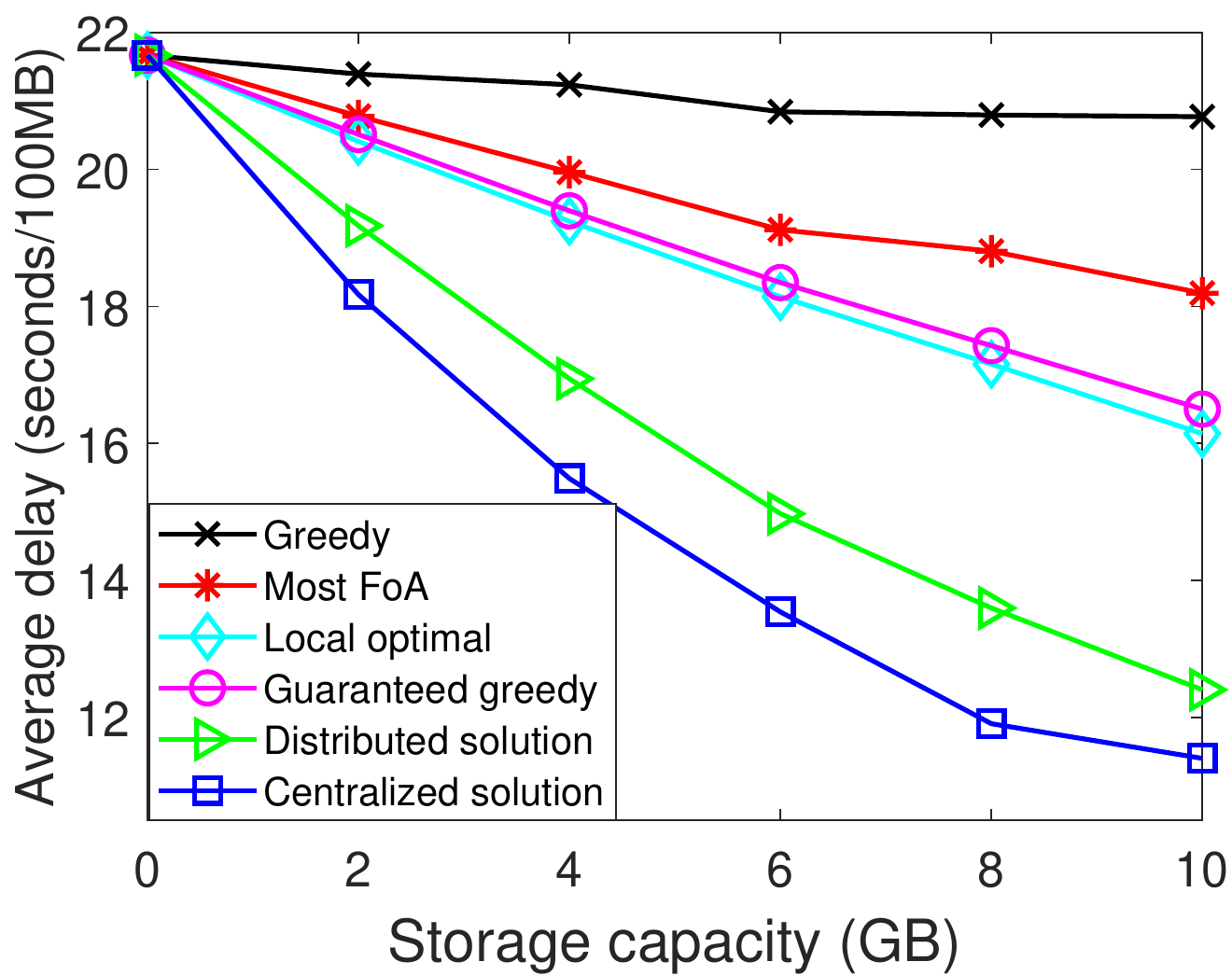} & 
		\hspace*{-.37cm}
		\epsfxsize=2.35 in \epsffile{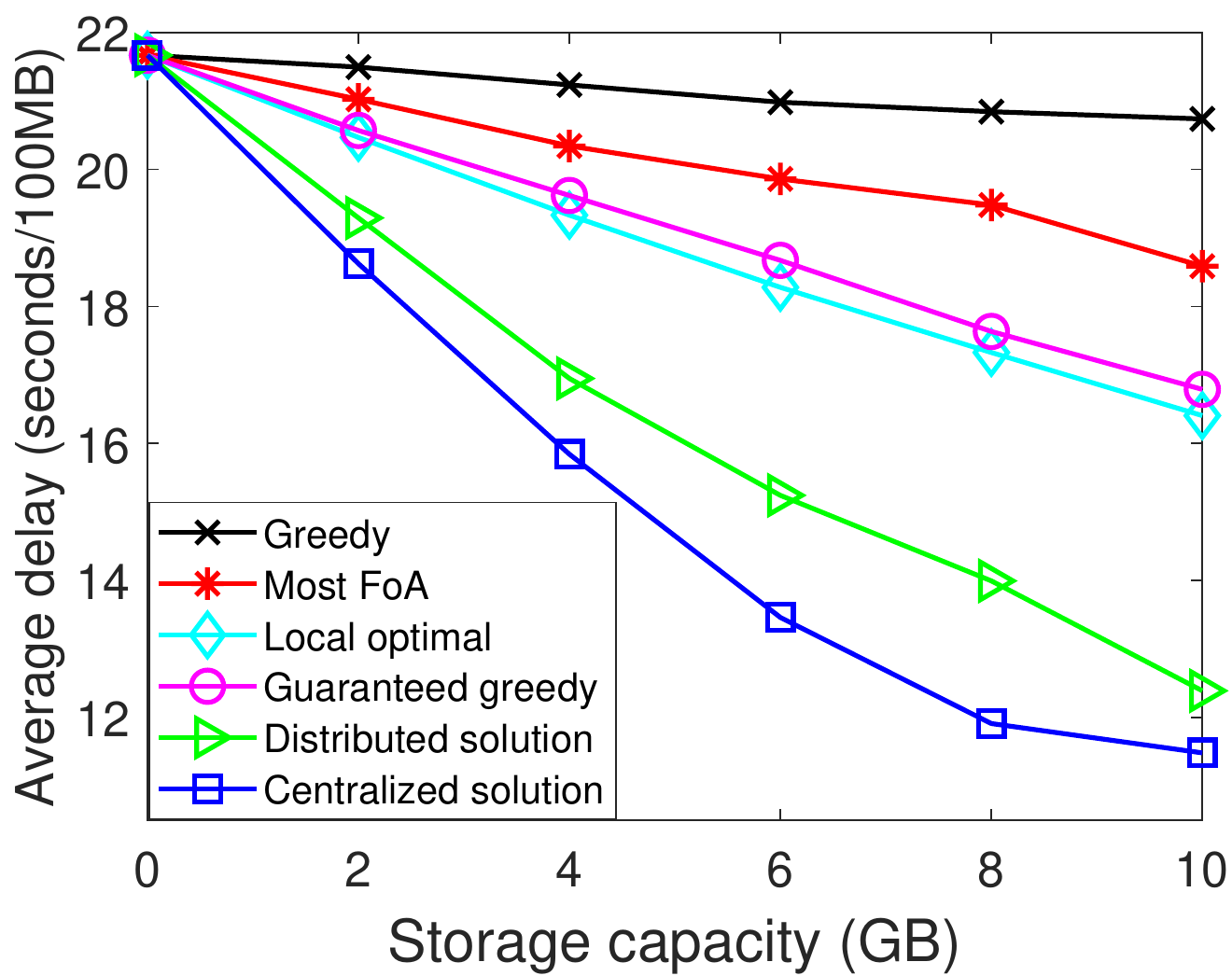} \\ [0.1cm]
		\text{\footnotesize (a) MEN-1} & \text{\footnotesize (b) MEN-2} & \text{\footnotesize (c) MEN-3} \\ [0.2cm]
		\end{array}$
		$\begin{array}{cc} 
		\epsfxsize=2.35 in \epsffile{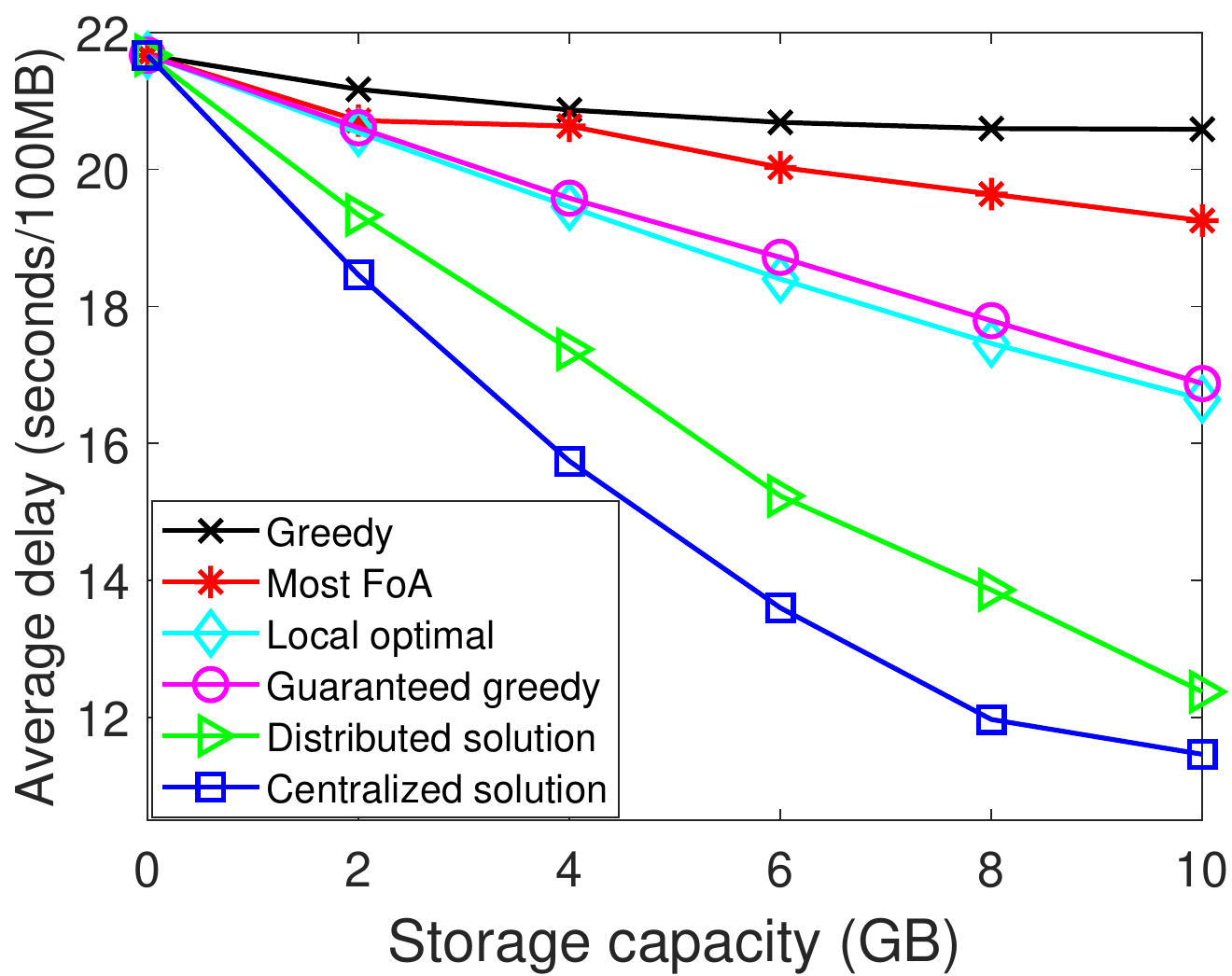} &
		\hspace*{-.37cm}
		\epsfxsize=2.4 in \epsffile{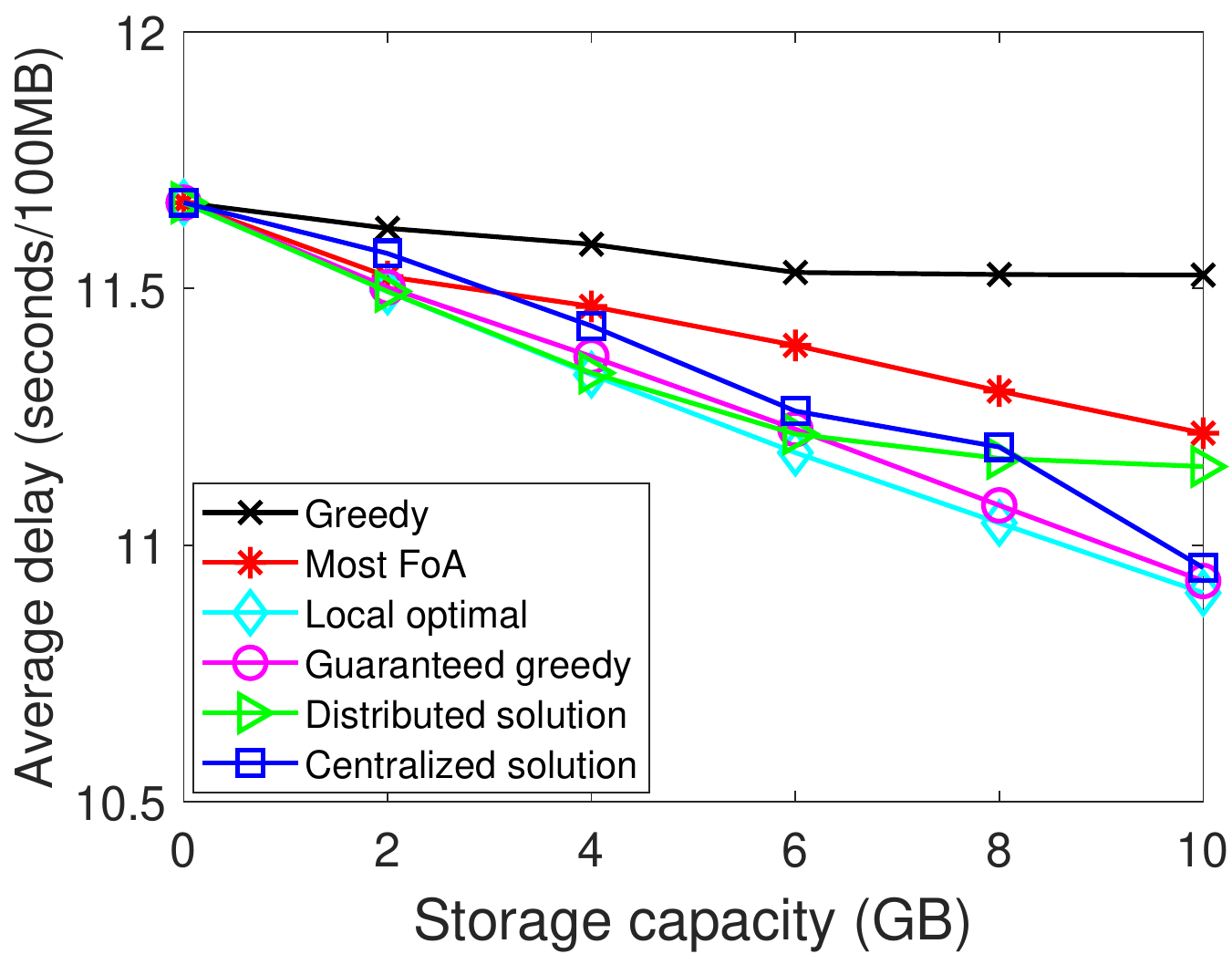} \\ [0.1cm]
		\text{\footnotesize (d) MEN-4} & \text{\footnotesize (e) MEN-5 (BS)} \\ [0cm]
		\end{array}$
		\caption{Average delay vs storage capacity for the MENs and the BS.}
		\label{fig:Avg_delay_inc_cap_pernode}
	\end{center}
\end{figure*}


Fig.~\ref{fig:Avg_delay_inc_cap} shows the trend of the total average delay in the MEC network with 200 considered available contents as the storage capacity increases from 0 to 10GB. It is observed that as the storage capacity increases, the average delays obtained by all policies are reduced, and the centralized solution outperforms all other policies. Specifically, when the storage capacity is 10GB, the average delay obtained by the centralized solution is 35\% lower than the locally optimal policy, 37\% lower than the guaranteed greedy policy, 52\% lower than the most FoA policy, and 66\% lower than the greedy policy. This is because the MENs can collaborate together to improve caching efficiency for the whole network when the centralized solution is used. In constrast, conventional policies (i.e., greedy, most FoA, guaranteed greedy, and locally optimal policies) do not consider cooperative caching among the MENs, and thus requested contents can only be downloaded from the CS (through the BS) if the contents are not stored at the associated MEN of the requesting users. For the distributed solution, although it cannot achieve the optimal solution, its performance is always better than the conventional policies, and the gap with the centralized solution is only within 12\%.  

We then observe the average delay from each participating MEN and the BS in Fig.~\ref{fig:Avg_delay_inc_cap_pernode}. In general, each MEN in Fig.~\ref{fig:Avg_delay_inc_cap_pernode}(a)-(d) follows the same trend as Fig.~\ref{fig:Avg_delay_inc_cap}. However, an interesting result can be seen in Fig.~\ref{fig:Avg_delay_inc_cap_pernode}(e). In our framework, we leverage the direct horizontal collaborations among MENs to minimize the average delay for the whole network. As such, since all nodes are also connected to the BS directly, the BS must sacrifice its performance to reduce the overall network delay by allowing other nodes to cache more popular contents. Nonetheless, even the BS sacrifices, its performance is still greater than those of the greedy and most FoA policies, and close to those of the locally optimal and guaranteed greedy policies.

\subsubsection{Effects of Number of Contents}

Fig.~\ref{fig:Avg_delay_inc_con} presents the total average delay as the number of contents increases from 50 to 200 contents. We fix the storage capacity for all MENs (including the BS) at 5GB. As expected, for a given storage capacity, if the number of contents increases, the average delays obtained by all policies increase. The reason is that the mobile users cannot download more contents from their directly connected MENs. Instead, the MENs need to download the contents from other nodes more frequently. Nevertheless, due to the collaboration, the average delay obtained by the centralized solution is still much lower by 15\% up to 55\% than those of all other conventional policies. Furthermore, the distributed solution can achieve the performance very close to that of the centralized solution especially when the number of contents is low. 

\begin{figure}[!]
	\centering
	\includegraphics[scale=0.5]{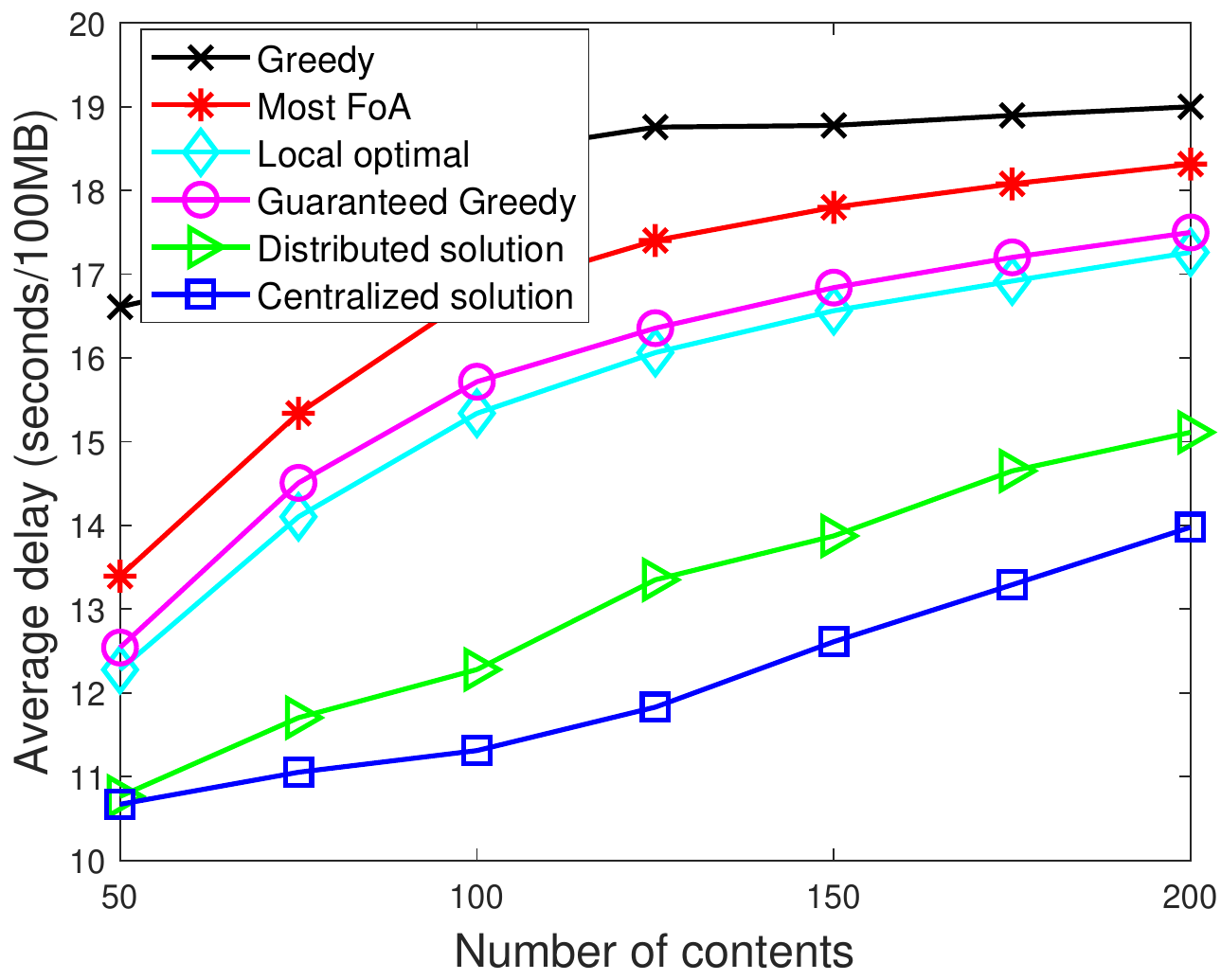}
	\caption{Average delay vs number of contents.}
	\label{fig:Avg_delay_inc_con}
\end{figure}

\subsubsection{Effect of Number of MENs}

Fig.~\ref{fig:Avg_delay_inc_en} shows the total average delay when we increase the number of MENs from 2 to 10. We fix the storage capacity of all MENs and the number of contents at 10GB and 200 contents, respectively. Interestingly, as the number of MENs increases, the total average delays obtained by the greedy, most FoA, guaranteed greedy, and locally optimal policies increase gradually. The reason is that these methods do not consider the collaboration among the MENs, and thus when the number of MENs increases, the delay at each MEN will contribute to the total delay of the network, yielding to an increasing trend. 

Nevertheless, the total average delay obtained by the distributed solution first dramatically decreases when the number of MENs increases from 2 to 4 (the same trend applies for the centralized solution). Then, the total average delay of the distributed solution slightly rises. However, it is still much lower by as much as 29\% compared with other conventional policies as the number of MENs increases. The reason of increasing trend in the total delay is that each MEN does not have full knowledge of the content caching decisions in the whole network. Thus, the MEN may download the content from the CS (through the BS) eventhough not directly connected MENs may cache the requested content. In contrast, due to the centralized control, the centralized solution can slightly reduce more delay and remain stable when the number of MENs exceeds 8. In this case, the total average delay of the centralized solution is at least 40\% lower than those of other conventional policies. Specifically, the centralized solution can utilize the efficiency of collaboration among the MENs, and thus as the number of MENs increases, more connections will be created, thereby better leveraging the collaboration among MENs. We also notice that the delay reduction of the centralized solution gets saturated after the number of MENs reaches a given number, indicating that all available contents are cached at MENs. Fig.~\ref{fig:Avg_delay_inc_en} also provides useful information to help the MEC service providers effectively deploy MENs.  

\begin{figure}[!]
	\centering
	\includegraphics[scale=0.5]{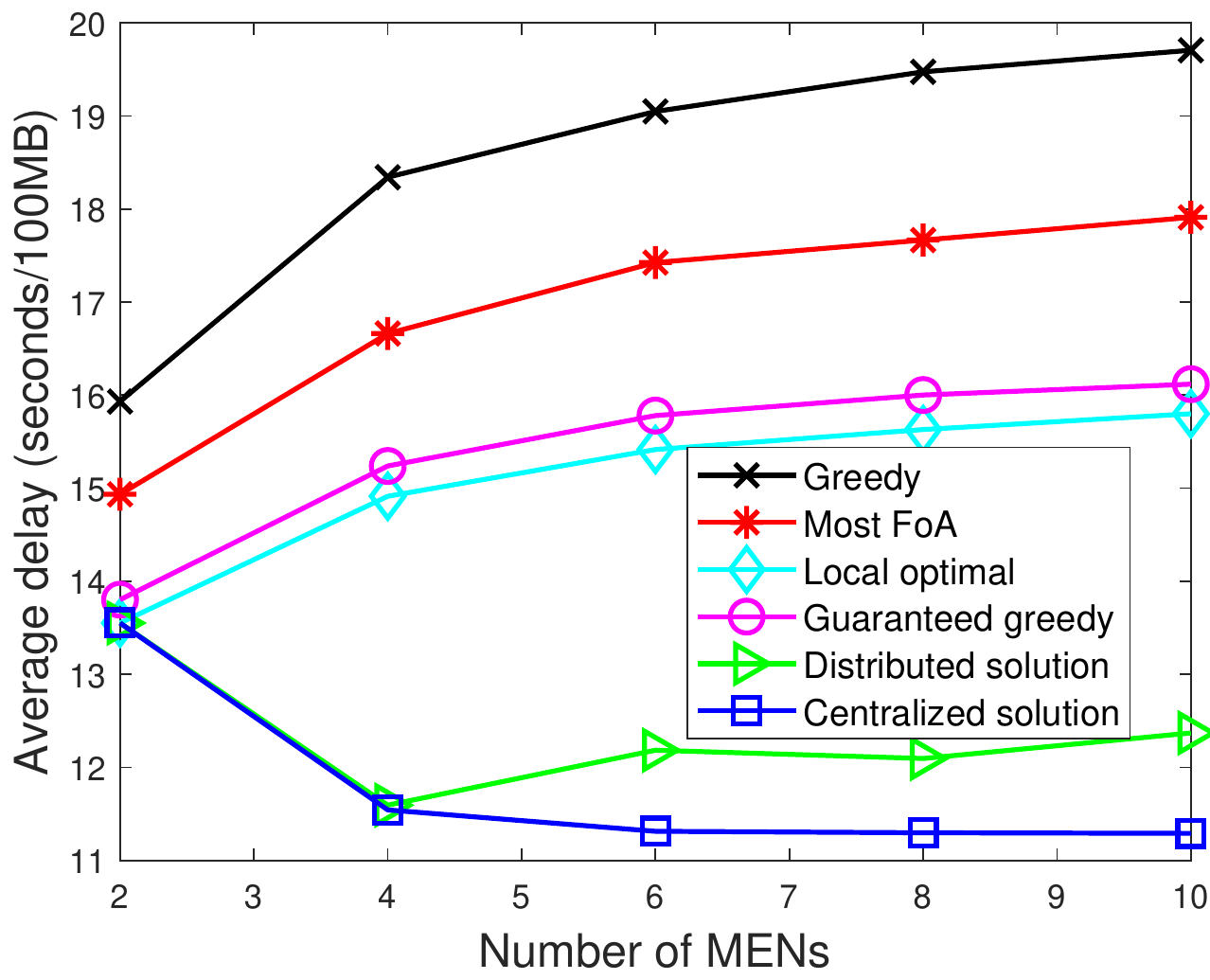}
	\caption{Average delay vs number of MENs.}
	\label{fig:Avg_delay_inc_en}
\end{figure}

\subsection{Cache Hit Rate Probability}

To further evaluate the performance of the proposed solutions, we conduct simulations to show the cache hit rate of the caching policies. Specifically, we generate incoming content requests randomly using a uniform distribution. In addition, we consider two types of the cache hit rate: (1) \emph{local-hit} and (2) \emph{global-hit}. The former represents the cache hit rate at each individual MEN, while the latter captures the cache hit rate from the whole network. Specifically, we denote $h_n$ as the local-hit of MEN-$n$. Then, the average cache hit rate for the whole network of other caching policies (i.e., greedy, most FoA, guaranteed greedy, and locally optimal policies) $h_{tot}$ is
\begin{equation}
\label{eqn:NR1}
\begin{aligned}
h_{tot} \defeq \frac{1}{N}\bigg[\sum_{n=1}^{N-1} h_n + N h_N\bigg],
\end{aligned}
\end{equation}
and of the proposed solutions $h^\ast_{tot}$ is
\begin{equation}
\label{eqn:NR2}
\begin{aligned}
h^\ast_{tot} \defeq \frac{1}{N}\sum_{n=1}^N \Big(h_n + \sum_{m \in \mathcal{M}_n^i, m\neq n} h_m\Big), \forall i \in \mathcal{I},
\end{aligned}
\end{equation}
where $h_N$ and ${{\underset{m \in \mathcal{M}_n^i, m\neq n}{\sum}}} h_m$ are the global-hit for the conventional policies and the proposed solutions, respectively.

\begin{figure*}[!]
\begin{center}
		 \hspace*{-.3cm}$\begin{array}{ccc} 
		\epsfxsize=2.05 in \epsffile{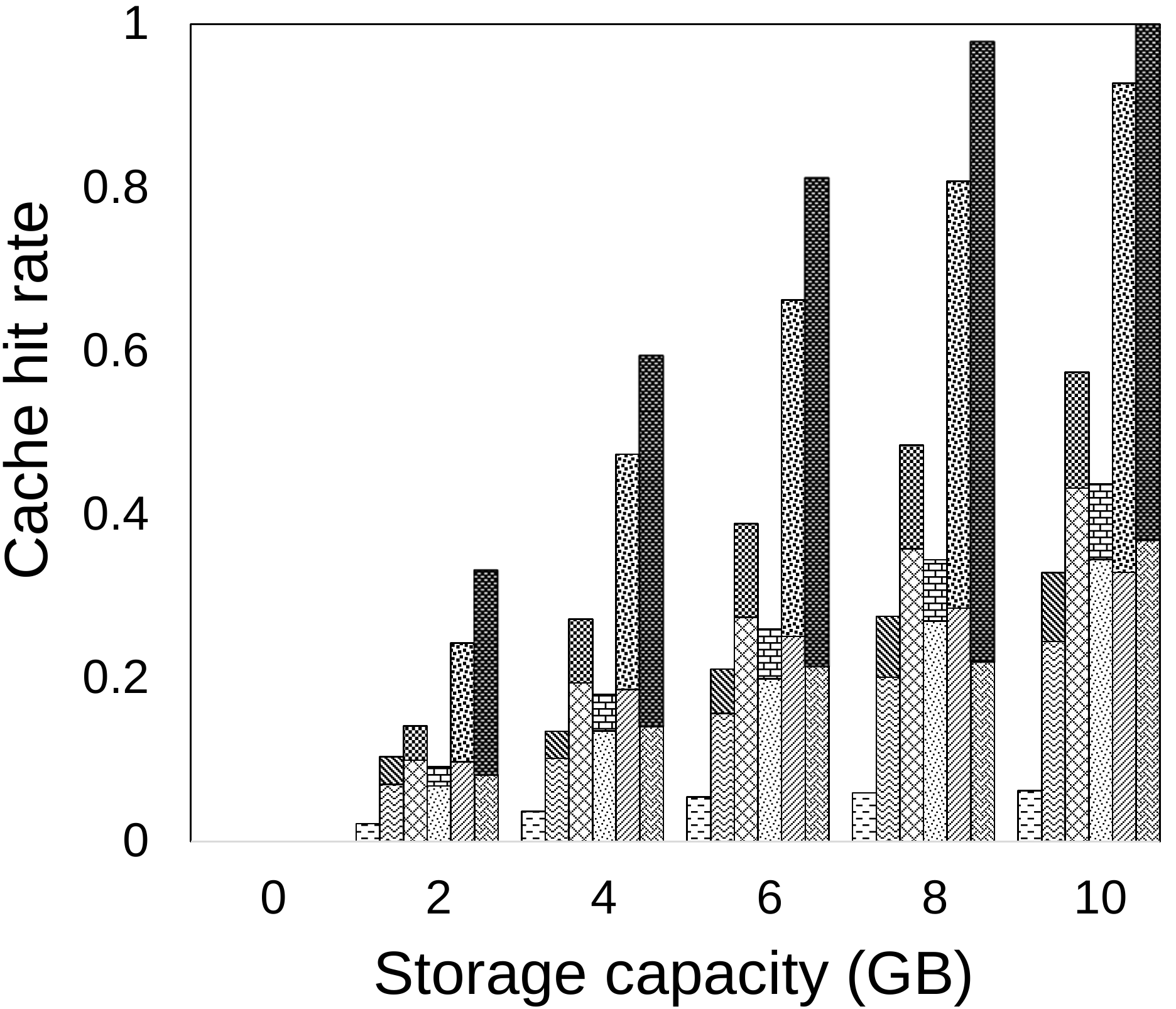} & 
		\epsfxsize=2.05 in \epsffile{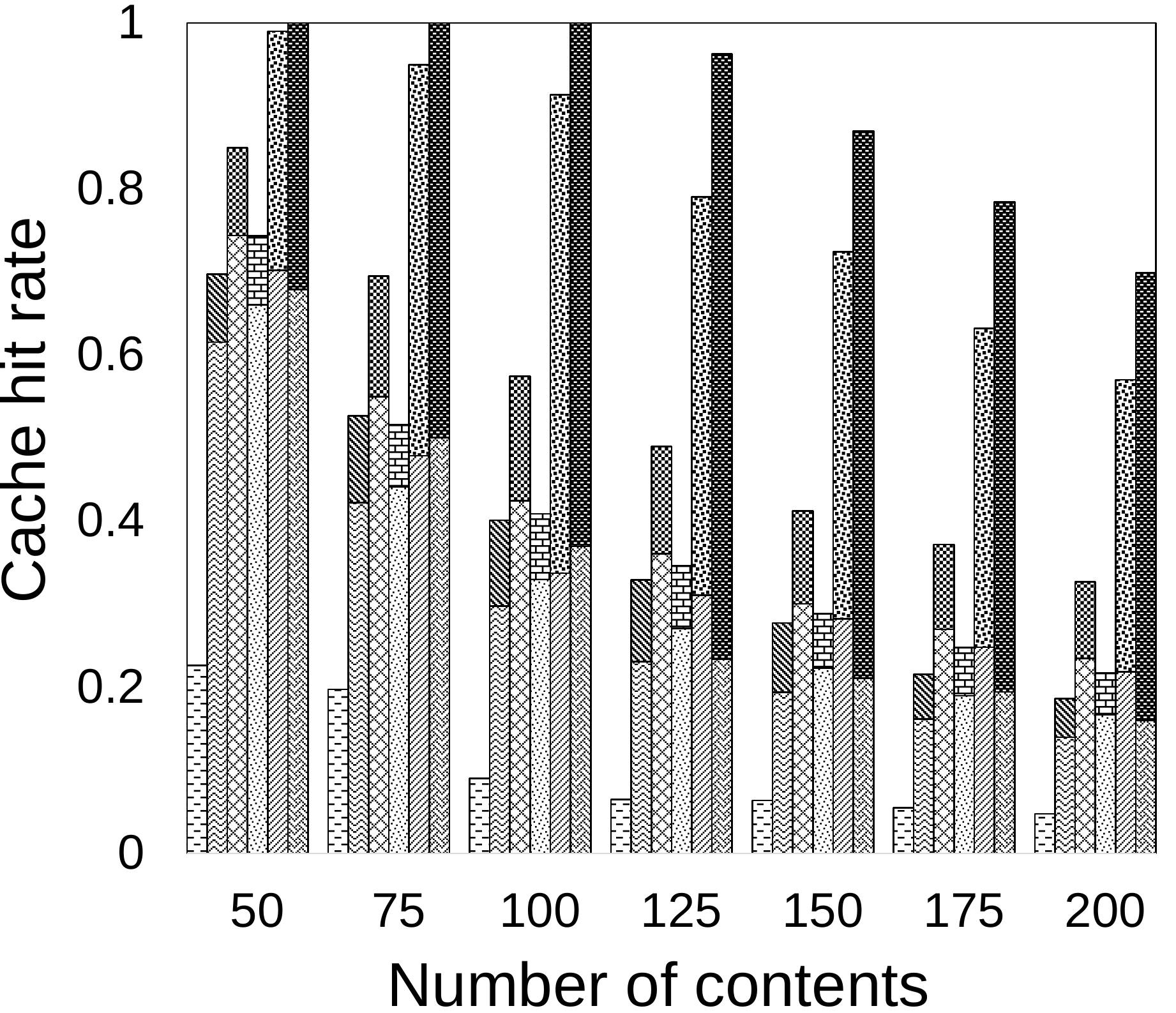} & 
	     \epsfxsize=2.47 in \epsffile{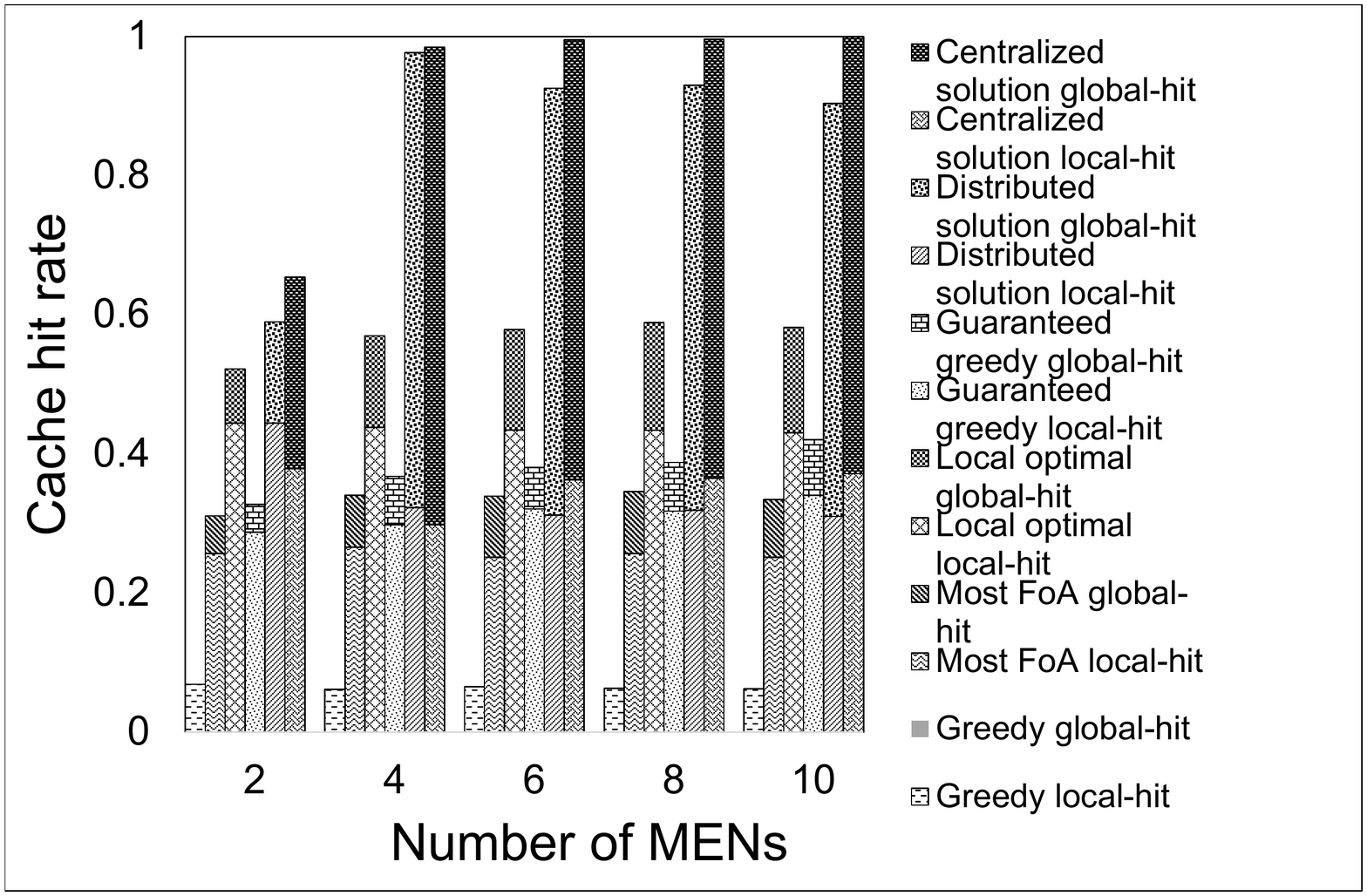} \\  [0cm]
		\text{\footnotesize (a)} & \text{\footnotesize (b)} & \text{\footnotesize (c)} \\
		\end{array}$
		\caption{Average cache hit rate when (a) the storage capacity increases, (b) the number of contents increases, and (c) the number of MENs increases}
		\label{fig:Avg_cache_hit}
	\end{center}
\end{figure*}

Fig.~\ref{fig:Avg_cache_hit} presents the average cache hit rates for all scenarios. Note that, for the greedy policy we only show the local-hit because the greedy policy is based on the sizes of contents only, and thus the MENs and the BS have the same set of cached contents. As observed in Fig.~\ref{fig:Avg_cache_hit}, the average cache hit rates of most FoA, guaranteed greedy, and locally optimal policies are greater than that of the greedy policy because they have less duplicate contents on the MENs. In particular, the MENs may have different users' demands, and thus they may cache dissimilar contents. For the distributed and centralized solutions, by leveraging the direct horizontal collaboration, their cache hit rates (including local-hit and global-hit) are 3 - 4 times greater than those of other policies. These figures also clearly show impacts of the collaborations among MENs. In particular, although the local-hit may not always be the best (because MENs may sacrifice to cache contents for other nodes to minimize the average delay for the whole network), the total cache hit rate for the network obtained by the centralized solution always achieves the highest value. It is also worth noting that when the total cache hit rate in the network increases, the traffic load on the backhaul network will be reduced.

%
%

\section{Conclusion}
\label{sec:Conc}
In this paper, we have introduced the effective joint cooperate caching and delivering framework (JOCAD) by leveraging direct horizontal cooperations among MENs. This framework aims to minimize the total average delay for the MEC network and lessen the network traffic on the backhaul network. To address the optimal joint caching and delivering problem, we have proposed the novel transformation method together with the improved branch-and-bound algorithm with the interior-point method. To reduce the complexity and communication overheads among MENs, we have also introduce distributed cooperative caching{-delivering} solution which minimizes the {duplicate contents} among directly connected MENs. Through the simulation results, we have shown that the proposed solutions can significantly outperform other caching policies in terms of the total average delay and cache hit rate. Furthermore, the results can provide useful knowledge for MEC service providers to tradeoff between the quality of service and the implementation costs in the MEC network. To further guarantee the performance, in the future work, we will develop effective distributed algorithms which can theoretically show the bounds and compare their performance with the current proposed solutions.



\ifCLASSOPTIONcaptionsoff
  \newpage
\fi

\vspace{0.7cm}

\clearpage
\appendices

\section{Proof of Lemma 1}
\label{appx:lemma1}
To prove this Lemma, we prove that one special case of the problem is NP-Complete. 
Specifically, in the problem ($\mathbf{P}_1$), let's not take into account the delivering problem. In particular, if MEN-$n$ does not have the requested content, it will download the content directly from the CS via the BS. Then, the problem ($\mathbf{P}_1$) can be formulated as follows:
\begin{equation}
\label{eqn:proof1b}
\begin{aligned}
\min_{\mathbf{x}} F(\mathbf{x}) = \min_{\mathbf{x}} \sum_{i=1}^I \sum_{n=1}^{N} f_n^i \Bigg[ \sum_{u=1}^{U_n} \bigg(x_n^i d_\alpha + (1 - x_n^i)d_{\delta}\bigg)\Bigg].
\end{aligned}
\end{equation}
Based on Eq.~(\ref{eqn:proof1b}), the problem becomes binary linear programming {(with $I \times N$ number of binary variables)} which was proved to be NP-complete in \cite{Karp:1972} and \cite{Karp2:1972}. Since the special case problem is at least as hard as the hardest problems in NP (i.e., NP-complete problem), the nested dual binary nonlinear programming problem ($\mathbf{P}_1$) is NP-hard.

\section{Proof of Theorem 1}
\label{appx:theorem1}
In the following, we will prove that if $\mathbf{x^*}$ is the optimal solution of ($\mathbf{P}_1$), then it is also the optimal solution of ($\mathbf{P}_2$) and vice versa. We first recall the inner minimization $F_{\emph{\mbox{IL}}}(\mathbf{x})$ in ($\mathbf{P}_1$) as follows:
\begin{equation}
\label{eqn:theorem1a}
\begin{aligned}
\underset{m}{\min}\Big([x_m^i + (1 - x_m^i)V] \frac{c_i}{l^m_n}\Big) \defeq \\
\underset{m}{\min}\Big(Q_n^i(x_1^i),\ldots,Q_n^i(x_m^i),\ldots,Q_n^i(x_{M_n^i}^i)\Big),
\end{aligned}
\end{equation}
where $Q_n^i(x_m^i) \defeq [x_m^i + (1 - x_m^i)V] \frac{c_i}{l^m_n}$ and $M_n^i$ is the total number of directly connected MENs containing content $i$ for MEN-$n$, with $\forall m \in \mathcal{M}_n^i \subset \mathcal{N}, \forall n \in \mathcal{N}, m \neq n, \forall i \in \mathcal{I}$. Then we can rewrite the objective function in Eq.~(\ref{eqn:form_problem1a}) into
\begin{equation}
\label{eqn:theorem1b}
\begin{aligned}
&\min_{\mathbf{x}} F(\mathbf{x}) = \min_{\mathbf{x}} \sum_{i=1}^I \sum_{n=1}^{N} f_n^i \Bigg[ \sum_{u=1}^{U_n} \bigg( \Omega \big(\mathbf{x}\big) +  \\
&\Phi \big(\mathbf{x}\big)\min_{m}\Big(Q_n^i(x_1^i),\ldots,Q_n^i(x_m^i),\ldots,Q_n^i(x_{M_n^i}^i)\Big)\bigg) \Bigg], 
\end{aligned}
\end{equation}
where $\mathbf{x} = (x_n^i, x_m^i)$ with $\forall n,m \in \mathcal{N}, m \neq n, \forall i \in \mathcal{I}$. Suppose that $\mathbf{x^*} = ({\hat x}_n^i,{\hat x}_m^i), \forall n, m \in \mathcal{N}, m \neq n, \forall i \in \mathcal{I}$ is the optimal solution of $F(\mathbf{x})$ in Eq.~(\ref{eqn:theorem1b}). Without loss of generality, $m=k_n^i$ will be selected as the MEN which represents the minimum delivering decision to download content $i$ for MEN-$n$ if the following condition satisfies:
\begin{equation} 
\label{eqn:theorem1g0}
\begin{aligned}
Q_n^i({\hat x}_{m=k_n^i}^i) \leq Q_n^i({\hat x}_{m\neq k_n^i}^i), \\
\forall n, m, k_n^i \in \mathcal{N}, m \neq n,k_n^i \neq n, \forall i \in \mathcal{I}.
\end{aligned}
\end{equation}
Hence, $\underset{\mathbf{x}}{\min}\text{ }F(\mathbf{x})$ can achieve the optimal objective function as follows:
\begin{equation} 
\label{eqn:theorem1g1}
\begin{aligned}
F(\mathbf{x^*}) &= \sum_{i=1}^I \sum_{n=1}^{N} f_n^i \Bigg[\sum_{u=1}^{U_n} \bigg(\Omega \big(\mathbf{x^*}\big) + \Phi \big(\mathbf{x^*}\big)Q_n^i({\hat x}_{m=k_n^i}^i)\bigg)\Bigg].
\end{aligned}
\end{equation}

Next, from MINLP problem ($\mathbf{P}_2$), we can also obtain the following expression:
\begin{equation}
\label{eqn:theorem1g}
\begin{aligned}
&\min_{\{\mathbf{x},\mathbf{y},\mathbf{z}\}} F(\mathbf{x},\mathbf{y},\mathbf{z}) = \\
&\min_{\{\mathbf{x},\mathbf{y},\mathbf{z}\}} \sum_{i=1}^I \sum_{n=1}^{N} f_n^i \Bigg[ \sum_{u=1}^{U_n} \bigg( \Omega \big(\mathbf{x}\big) +  \Phi \big(\mathbf{x}\big)z_n^i\bigg) \Bigg],
\end{aligned}
\end{equation}
\begin{eqnarray}
\text{s.t. (\ref{eqn:prop_sol2c})-(\ref{eqn:prop_sol2f})}. \nonumber
\end{eqnarray}
Given the optimal solution $\mathbf{x^*} = ({\hat x}_n^i,{\hat x}_m^i), \forall n, m \in \mathcal{N}, m \neq n, \forall i \in \mathcal{I}$ from the Eq.~(\ref{eqn:theorem1b}), we can rewrite the Eq.~(\ref{eqn:theorem1g}) by:
\begin{equation}
\label{eqn:theorem1g2}
\begin{aligned}
&\min_{\{\mathbf{x^*},\mathbf{y},\mathbf{z}\}} F(\mathbf{x^*},\mathbf{y},\mathbf{z}) = \\ 
&\min_{\{\mathbf{x^*},\mathbf{y},\mathbf{z}\}} \sum_{i=1}^I \sum_{n=1}^{N} f_n^i \Bigg[ \sum_{u=1}^{U_n} \bigg( \Omega \big(\mathbf{x^*}\big) +  \Phi \big(\mathbf{x^*}\big)z_n^i\bigg) \Bigg], \quad \\ 
&\text{s.t.} \quad
\left\{	\begin{array}{ll}
z_n^i \geq Q_n^i({\hat x}_m^i) - V y_m^i, \forall n,m \in \mathcal{N}, m\neq n,\forall i \in \mathcal{I}, \\
\sum\limits_{m \in \mathcal{M}_n^i} y_m^i = M_n^i - 1, \forall n \in \mathcal{N},\forall i \in \mathcal{I},\\
y_m^i \in \{0,1\}, \forall m \in \mathcal{N}, \forall i \in \mathcal{I},\\
z_n^i \in \mathbb{R}_0^+, \forall n \in \mathcal{N}, \forall i \in \mathcal{I}.
\end{array}	\right.\\
\end{aligned}
\end{equation}
Then, consider all feasible values of $y_m^i$ in Eq.~(\ref{eqn:theorem1g2}). Based on the constraints $\sum\limits_{m \in \mathcal{M}_n^i} y_m^i = M_n^i - 1$, there are one MEN-$m$ with $y_m^i = 0$ and $M_n^i-1$ MENs with $y_m^i = 1$. Suppose that $m = \kappa_n^i$, where $\forall \kappa_n^i \in \mathcal{M}_n^i,\forall n \in \mathcal{N},\forall i \in \mathcal{I}$, is the current MEN which has $y_{m=\kappa_n^i}^i = 0$. Thus, $z_n^i \geq Q_n^i({\hat x}_{m=\kappa_n^i}^i)$ for $m = \kappa_n^i$ and $z_n^i \geq Q_n^i({\hat x}_{m\neq \kappa_n^i}^i) - V$ for $m \neq \kappa_n^i$. Since $V$ is a very big value, we can eliminate the rest of the constraints when $m \neq \kappa_n^i$. As a result, only $z_n^i \geq Q_n^i({\hat x}_{m=\kappa_n^i}^i), \forall n \in \mathcal{N},\forall i \in \mathcal{I}$ is applied for each possible optimization as follows:
\begin{equation}
\label{eqn:theorem1h}
\begin{aligned}
&\min_{\{\mathbf{x^*},\mathbf{z}\}} F(\mathbf{x^*},\mathbf{z}) = \\
&\min_{\{\mathbf{x^*},\mathbf{z}\}}\sum_{i=1}^I \sum_{n=1}^{N} f_n^i \Bigg[ \sum_{u=1}^{U_n} \bigg( \Omega \big(\mathbf{x^*}\big) +  \Phi \big(\mathbf{x^*}\big)z_n^i\bigg) \Bigg], \\
&\text{s.t.} \quad z_n^i \geq Q_n^i({\hat x}_{m=\kappa_n^i}^i), \forall n \in \mathcal{N},\forall i \in \mathcal{I}.
\end{aligned}
\end{equation}
If $\mathbf{z^*} = {\hat z}_n^i$, $\forall n \in \mathcal{N}, \forall i \in \mathcal{I}$ is also the optimal solution, then the possible optimal objective value $F(\mathbf{x^*},\mathbf{z^*})$ for each $\kappa_n^i$ is
\begin{equation}
\label{eqn:theorem1h2}
F(\mathbf{x^*},\mathbf{z^*}) = \sum_{i=1}^I \sum_{n=1}^{N} f_n^i \Bigg[ \sum_{u=1}^{U_n} \bigg( \Omega \big(\mathbf{x^*}\big) +  \Phi \big(\mathbf{x^*}\big){\hat z}_n^i\bigg) \Bigg], 
\end{equation}
where ${\hat z}_n^i = Q_n^i({\hat x}_{m=\kappa_n^i}^i),\forall n \in \mathcal{N},\forall i \in \mathcal{I}$. Therefore, we can derive:
\begin{equation}
F(\mathbf{x^*}) = \sum_{i=1}^I \sum_{n=1}^{N} f_n^i \Bigg[ \sum_{u=1}^{U_n} \bigg( \Omega \big(\mathbf{x^*}\big) +  \Phi \big(\mathbf{x^*}\big)Q_n^i({\hat x}_{m=\kappa_n^i}^i)\bigg) \Bigg].
\end{equation}
Based on Eq.~(\ref{eqn:theorem1g0}), from each MEN-$n$ we can choose an MEN-$m$ to download content $i$ with minimum delivering time, i.e., $Q_n^i({\hat x}_{m=k_n^i}^i) \leq Q_n^i({\hat x}_{m\neq k_n^i}^i)$. In other words, the selected node $\kappa_n^i$ in $(\mathbf{P_2})$ is the same as the selected node $k_n^i$ in $(\mathbf{P_1})$, $\forall n \in \mathcal{N}$ and $\forall i \in \mathcal{I}$. As a result, if $\mathbf{x^*}$ is the optimal solution of $(\mathbf{P_1})$, it is also the optimal solution of $(\mathbf{P_2})$.


Similarly, if ($\mathbf{x^*},\mathbf{y^*},\mathbf{z^*}$) is the optimal solution of ($\mathbf{P}_2$), we can prove that $\mathbf{x^*}$ is also the optimal solution of ($\mathbf{P}_1$). We recall MINLP problem ($\mathbf{P}_2$) in Eq.~(\ref{eqn:theorem1g}). Given that $\mathbf{x^*} = ({\hat x}_n^i,{\hat x}_m^i)$, $\mathbf{y^*} = {\hat y}_m^i$, and $\mathbf{z^*} = {\hat z}_n^i, \forall n, m \in \mathcal{N}, m \neq n, \forall i \in \mathcal{I}$ are the optimal solutions of ($\mathbf{P}_2$) such that
\begin{equation}
\label{eqn:theorem1k}
\begin{aligned}
{\hat z}_n^i = Q_n^i({\hat x}_{m=\kappa_n^i}^i),\forall n \in \mathcal{N},\forall i \in \mathcal{I},
\end{aligned}
\end{equation}
where $m = \kappa_n^i, \forall n \in \mathcal{N},\forall i \in \mathcal{I}$ is the current MEN which has ${\hat y}_{m=\kappa_n^i}^i = 0$. After ignoring other MENs when $m \neq \kappa_n^i, \forall n \in \mathcal{N},\forall i \in \mathcal{I}$ because of ${\hat y}_{m \neq \kappa_n^i}^i = 1$ and the influence of big value $V$, then we have
\begin{equation}
\label{eqn:theorem1l}
F(\mathbf{x^*},\mathbf{z^*}) = \sum_{i=1}^I \sum_{n=1}^{N} f_n^i \Bigg[ \sum_{u=1}^{U_n} \bigg( \Omega \big(\mathbf{x^*}\big) +  \Phi \big(\mathbf{x^*}\big){\hat z}_n^i\bigg) \Bigg], 
\end{equation}
where ${\hat z}_n^i = Q_n^i({\hat x}_{m=\kappa_n^i}^i),\forall n \in \mathcal{N},\forall i \in \mathcal{I}$, and thus
\begin{equation}
F(\mathbf{x^*}) = \sum_{i=1}^I \sum_{n=1}^{N} f_n^i \Bigg[ \sum_{u=1}^{U_n} \bigg( \Omega \big(\mathbf{x^*}\big) +  \Phi \big(\mathbf{x^*}\big)Q_n^i({\hat x}_{m=\kappa_n^i}^i)\bigg) \Bigg].
\end{equation}

From ($\mathbf{P}_1$), we also have $m=k_n^i$, where $\forall k_n^i \in \mathcal{M}_n^i,\forall n \in \mathcal{N},\forall i \in \mathcal{I}$, as the selected MEN which indicates the minimum delivering decision to download content $i$ for MEN-$n$ if
\begin{equation}
\label{eqn:theorem1m}
\begin{aligned}
Q_n^i({\hat x}_{m=k_n^i}^i) \leq Q_n^i({\hat x}_{m\neq k_n^i}^i), \\
\forall n, m, k_n^i \in \mathcal{N}, m \neq n, k_n^i \neq n, \forall i \in \mathcal{I}.
\end{aligned}
\end{equation}
Considering the condition in Eq.~(\ref{eqn:theorem1k}) and Eq.~(\ref{eqn:theorem1m}), there exists MENs such that $m = \kappa_n^i = k_n^i, \forall n \in \mathcal{N}$ and $\forall i \in \mathcal{I}$. Hence, without loss of generality, we can achieve the final optimal objective value of ($\mathbf{P}_1$) as follows:
\begin{equation}
\label{eqn:theorem1n}
\begin{aligned}
&F(\mathbf{x^*}) = \sum_{i=1}^I \sum_{n=1}^{N} f_n^i \Bigg[ \sum_{u=1}^{U_n} \bigg( \Omega \big(\mathbf{x^*}\big) +  \Phi \big(\mathbf{x^*}\big)Q_n^i({\hat x}_{m=\kappa_n^i}^i)\bigg) \Bigg] \quad \\
&= \sum_{i=1}^I \sum_{n=1}^{N} f_n^i \Bigg[ \sum_{u=1}^{U_n} \bigg( \Omega \big(\mathbf{x^*}\big) +  \Phi \big(\mathbf{x^*}\big)Q_n^i({\hat x}_{m=k_n^i}^i)\bigg) \Bigg],
\end{aligned}
\end{equation}
which is the same as $F(\mathbf{x^*},\mathbf{z^*})$ in Eq.~(\ref{eqn:theorem1l}). Thus, we can conclude that ($\mathbf{P}_1$) is equivalent to ($\mathbf{P}_2$).

\section{Proof of Theorem 2}\label{appx:theorem2}
We adopt this proof from~\cite{Smith:1979} and~\cite{Smith:1984}. Suppose that $T_{\rho}$ specifies the expected number of total average delays searched in the depth $\rho$. Then, we obtain that $T_0=1$ at root problem ($\mathbf{RP}$) and $T_{\rho} \geq 1$ for $\rho \geq 1$, i.e., at subproblems ($\mathbf{SP}$s). Thus, we can derive $T$ by
\begin{equation}
\label{eqn:prop_sol7a}
\begin{aligned}
T = 1 + \sum_{\rho=1}^\varrho \geq 1 + \sum_{\rho=1}^\varrho 1 \geq 1 + \varrho.
\end{aligned}
\end{equation}
Furthermore, there exists a constant $\varsigma$ such that $T_{\varrho} \leq \varsigma \leq T_{\infty}$ for all $\varrho$, where $T_{\infty} = \underset{\varrho\to\infty}{\lim}T_{\varrho}$, and $T_0 \leq T_1 \leq \ldots \leq T_\varrho \leq \varsigma$. Then, we can express that
\begin{equation}
\label{eqn:prop_sol7b}
\begin{aligned}
T \leq 1 + \sum_{\rho=1}^\varrho \varsigma \leq 1 + \varsigma\varrho.
\end{aligned}
\end{equation}
From Eq.~(\ref{eqn:prop_sol7a}) and Eq.~(\ref{eqn:prop_sol7b}), it implies that $1+\varrho \leq T \leq 1 + \varsigma\varrho$. Alternatively, we can state that $T$ is linear (or polynomial) at the depth $\varrho$.


\section{Proof of Theorem 3}\label{appx:theorem4}
We adopt this proof from~\cite{Fowkes:2013}. Given that the $t$-th total average delay $\Psi_t = \text{arg }\underset{t}{\min}\text{ }  F_{\mathbf{SP}}(\mathbf{x},\mathbf{y},\mathbf{z})$, and $\psi_t$ is a search region to obtain $\Psi_t$ for some $\theta_t < t$. Then, there exists $\phi > 0$ such that for $\eta > 0$ and any $\mathbf{x}$, $\mathbf{y}$, $\mathbf{z}$, we have
\begin{equation}
\label{eqn:prop_sol13}
\begin{aligned}
F_{\mathbf{SP}}(\mathbf{x},\mathbf{y},\mathbf{z}) \leq \phi \implies \beta_U(\Psi_t) - \beta_L(\Psi_t) \leq \eta.
\end{aligned}
\end{equation}
Then, $\psi_t$ should also have $F_{\mathbf{SP}}(\mathbf{x},\mathbf{y},\mathbf{z}) \leq \phi$ when $t = {\hat t} \in \mathbb{N}$. Based on Eq.~(\ref{eqn:prop_sol13}), we can obtain
\begin{equation}
\label{eqn:prop_sol14}
\begin{aligned}
\beta_U(\psi_{\hat t}) - \beta_L(\psi_{\hat t}) \leq \eta. 
\end{aligned}
\end{equation}
Since $\psi_{\hat t}$ is split (to create two new subproblems) at step $\theta_{\hat t}$, then $\beta_L(\psi_{\hat t}) = \beta_L^{(\theta_{\hat t})}$. As a result, the condition becomes
\begin{equation}
\label{eqn:prop_sol15}
\begin{aligned}
\beta_U^{(\theta_{\hat t})} - \beta_L^{(\theta_{\hat t})} \leq \beta_U(\psi_{\hat t}) -  \beta_L^{(\theta_{\hat t})} \leq \eta, 
\end{aligned}
\end{equation}
and thus $\beta_U^{(\theta_{\hat t})} \leq \beta_U(\psi_{\hat t})$. Next, consider that the final total average delay ${\hat \Psi}$ is obtained at $\mathbf{x}^*$, $\mathbf{y}^*$, and $\mathbf{z}^*$. To satisfy a condition that $\beta_L^{(\theta_{\hat t})} \leq {\hat \Psi}$, the following expression can be obtained based on Eq.~(\ref{eqn:prop_sol15}):
\begin{equation}
\label{eqn:prop_sol16}
\begin{aligned}
\beta_U^{(\theta_{\hat t})} - {\hat \Psi} \leq \beta_U^{(\theta_{\hat t})} -  \beta_L^{(\theta_{\hat t})} \leq \eta,
\end{aligned}
\end{equation}
where $\beta_U^{(\theta_{\hat t})}$ is shown within the optimality tolerance $\eta$ of the $F(\mathbf{x},\mathbf{y},\mathbf{z})$.

\end{document}